\documentclass[11pt]{article}
\usepackage{amsmath,amssymb,amsfonts,amsthm,epsfig,verbatim}
\usepackage{times}
\usepackage{color}
\newif\ifhyper\IfFileExists{hyperref.sty}{\hypertrue}{\hyperfalse}
\hypertrue
\ifhyper\usepackage{hyperref}\fi
\usepackage{microtype}
\DisableLigatures{encoding = *, family = *}
\newtheorem{theorem}{Theorem}

\newtheorem{lemma}[theorem]{Lemma}
\newtheorem{proposition}[theorem]{Proposition}
\newtheorem{corollary}[theorem]{Corollary}
\newtheorem{claim}[theorem]{Claim}
\newtheorem{fact}[theorem]{Fact}

\newtheorem{definition}[theorem]{Definition}
\newtheorem{conjecture}[theorem]{Conjecture}
\newtheorem{remark}[theorem]{Remark}

\newcommand{\etal}{{\em et al.\ }}

\newcommand{\braket}[1]{\langle #1 \rangle}

\newcommand{\Cov}{\operatorname{Cov}}

\newcommand{\ignore}[1]{}

\newcommand{\cref}[1]{Corollary~\ref{cor:#1}}

\newcommand{\R}{{\mathbb{R}}}

\newcommand{\E}{\operatorname{{\bf E}}}

\newcommand{\littlesum}{\mathop{\textstyle \sum}}

\newcommand{\Inf}{\mathrm{Inf}}

\renewcommand{\Pr}{\operatorname{{\bf Pr}}}

\newcommand{\newa}[1]{{{#1}}}

\makeatletter
\renewcommand{\subsection}{\@startsection{subsection}{2}{0pt}{-6pt}{-5pt}{\normalsize\bf}}
\renewcommand{\subsubsection}{\@startsection{subsubsection}{3}{0pt}{-12pt}{-5pt}{\normalsize\bf}}
\makeatother

\date{}

\begin{document}

\title{Explicit optimal hardness via Gaussian stability results}

\author{Anindya De\thanks{{\tt anindya@cs.berkeley.edu}.  Research supported by Satish Rao's NSF award  CCF-1118083.}\\ Computer Science Division, \\
University of California, Berkeley\\
\and
Elchanan Mossel\thanks{{\tt mossel@stat.berkeley.edu}. Research supported by NSF award DMS-1106999, CCF 1320105
and DOD ONR grant N000141110140  }\\ Dept. of Statistics and Computer Science, \\
University of California, Berkeley\\ 
}

\maketitle

\setcounter{page}{0}

\thispagestyle{empty}

~
\vskip -.5in
~

\begin{abstract}
The results of Raghavendra (2008) show that assuming Khot's Unique Games Conjecture (2002), for every constraint satisfaction problem there exists a generic semi-definite program that achieves the optimal approximation factor. This result is existential as it does not provide an explicit optimal rounding procedure nor does it allow to calculate exactly the Unique Games hardness of the problem. 

Obtaining an explicit optimal approximation scheme and the corresponding approximation factor is a difficult challenge for each specific approximation problem. {Khot \etal (2004) established a general approach for determining the exact approximation factor and the corresponding optimal rounding algorithm for any given constraint satisfaction problem.} 
However, this approach crucially relies on results explicitly proving optimal partitions in the Gaussian space. Until recently, Borell's result (1985) was the only non-trivial Gaussian partition 
result  known. 

In this paper we derive the first explicit optimal approximation algorithm and the corresponding approximation factor using a new result on Gaussian partitions due to Isaksson and Mossel (2012).  This Gaussian result allows us to determine the exact Unique Games Hardness of MAX-3-EQUAL. In particular, our results show that Zwick's algorithm for this problem achieves the optimal approximation factor and prove that the approximation achieved by the algorithm is $\approx 0.796$ as conjectured by Zwick. 

We further use the previously known optimal Gaussian partitions results to obtain a new Unique Games Hardness factor for MAX-k-CSP:
Using the well known fact that jointly normal pairwise independent random variables are fully independent, we show that the UGC hardness of Max-k-CSP is $\frac{\lceil (k+1)/2 \rceil}{2^{k-1}}$, improving on results of Austrin and Mossel (2009).
\end{abstract}

\maketitle

\section{Introduction}
The study of inapproximability of Constraint Satisfaction Problems (CSPs) has been an important area of research in complexity theory in the past two decades. A CSP is specified by a alphabet $[q]$ and a set of predicates $\mathcal{P}$ such that all $P \in \mathcal{P} : [q]^k \rightarrow \{0,1\}$\footnote{We are assuming a somewhat restricted form of a CSP where all the predicates have the same arity.}. Here $k$ is called the arity of the predicate. An instance of the problem (say $G$)  is given by $n$ variables $x_1, \ldots, x_n$ and a set of constraints $\mathcal{E}$ such that every $e \in \mathcal{E}$ is of the form $e= (S, P)$ where $S \in [n]^k$ and $P \in \mathcal{P}$.

Now, consider any mapping $\mathcal{L} : [n] \rightarrow [q]$.  A constraint $e =(S,P)$ is said to be ``satisfied" if $P(\mathcal{L}(S_1), \ldots, \mathcal{L}(S_k))=1$ where $S_i$ is the $i^{th}$ element of $S$. We also define $val_{\mathcal{L}}(G)$ as {$ val_{\mathcal{L}}(G)= \mathbf{E}_{v \in \mathcal{E}} [P(\mathcal{L}(S_1), \ldots, \mathcal{L}(S_k))]$}. 
The algorithmic task is to come up with the mapping $\mathcal{L}$ such that $val_{\mathcal{L}}(G)$ is maximized. Towards this, we define $val(G) = \max_{\mathcal{L}} val_{\mathcal{L}}(G)$. 

The reason for studying the very general framework of CSPs is because many specific problems of interest say MAX-CUT, MAX-3-SAT etc.~fall in this framework. 
In the past two decades, there have been  important results in the study of inapproximability of CSPs including the monumental work of H{\aa}stad~\cite{Hastad} who obtained optimal inapproximability results for CSPs like MAX-3-SAT and MAX-3-LIN. Still, a gap continued to exist between the known algorithms and hardness results  for many important  CSPs like MAX-CUT and MAX-2-SAT. Towards closing this gap, Khot~\cite{Kho:02} introduced the Unique Games Conjecture (UGC) which stated the following (equivalent form from \cite{KKMO07}): 
\begin{conjecture}
Given any $\delta>0$, there is a prime $p$ such that given a set of linear equations $x_i-x_j  =c_{ij} \ (mod \ p)$, it is NP-hard to decide which one of the following is true:
\begin{itemize}
\item There is an assignment to the $x_i$'s which satisfies at least $1-\delta$ fraction of the constraints. 
\item All assignments to the $x_i$'s can satisfy at most $\delta$ fraction of the constraints. 
\end{itemize}
\end{conjecture}\label{conj:1}
A series of (often optimal) inapproximability results were proven using the Unique Games Conjecture starting with~\cite{KhotRegev:07,KKMO07} which culminated in the beautiful result of Raghavendra~\cite{Rag:08} who showed that for every CSP of constant arity and alphabet size, there is a simple and generic SDP which is optimal assuming the Unique Games Conjecture. More specifically, he showed the following. 
\begin{theorem}
Suppose that for the generic SDP, there is an instance $G$ such that $val(G)  = s$  while the SDP objective value is $c$. Then, assuming the UGC, given an instance $G'$ of the CSP such that $val(G') = c-\eta$, it is NP-hard to find a $\mathcal{L}$ such that $val_{\mathcal{L}}(G) \ge s +\eta$ for any $\eta>0$.  Further, there is an efficient rounding algorithm such that given an instance $G$ with value $c$ on the instance $G$, it finds an assignment $\mathcal{L}$ with value $s-\eta$ (for $\eta>0$). \end{theorem}

While this result essentially settles the question of approximability of CSPs from an an abstract perspective, perhaps not too surprisingly , it says nothing about the exact hardness factors for specific CSPs. This is in contrast to the situation in the case of MAX-CUT~\cite{KKMO07} or MAX-2-SAT~\cite{Austrin07} where exact inapproximability factors are known.
The reason is that in Raghavendra's framework (and all previous results), determining the optimal inapproximability result for a specific CSP requires knowledge of the optimal partitioning of the Gaussian space for the corresponding predicate. While the optimal partitioning is known for the predicates corresponding to MAX-CUT and MAX-2-SAT, it is not known 
for arbitrary predicates. In fact, it should also be mentioned that while Raghavendra's result is a generalization of the results for  MAX-CUT and MAX-2-SAT, it does not imply the results for MAX-CUT or MAX-2-SAT without the knowledge of the optimal Gaussian partitioning. Likewise, even though the rounding algorithm in~\cite{Rag:08} is efficient, it is a brute force search over a small space that results only in a close to optimal rounding scheme. Thus, in a sense, the result provides implicitly a sequence  of rounding algorithm whose approximation factors is guaranteed to converge to the hardness factor. This again is different from the rounding algorithms in \cite{GW95,Zwick98,LLZ} where the rounding algorithm is far more explicit (in the first two cases, it is simply random hyperplane rounding). 

We now elaborate on the reason for difficulty in establishing exact hardness factors: The exact hardness factor in the case of MAX-CUT~\cite{KKMO07} and MAX-2-SAT~\cite{Austrin07} crucially rely on Gaussian Analysis. More specifically, it uses the invariance principle~\cite{MOO:10} together with a result in Gaussian space specifying explicitly an {\em Optimal Gaussian Partition} for the particular predicate.  However, only few optimal Gaussian partitions are known (or even conjectured).  In fact, to the best of our knowledge, before this paper, Borell's result \cite{Borell:85} was the only non-trivial Gaussian partition result used in hardness of approximation (for e.g.,  \cite{KKMO07,Austrin07}).  

The above issue also explains the ``brute-force" search aspect of the rounding scheme in \cite{Rag:08}. The optimal rounding scheme and the optimal gaussian partitioning (for a given predicate) are known to be intimately linked to each other (see \cite{Rag:08} for a detailed explanation). In absence of knowledge of the optimal partitioning, {\cite{Rag:08} uses dimension reduction (\cite{JohnsonLindenstrauss:84}) to reduce the dimension of the SDP solution and subsequently resorts to brute-force search in the low-dimensional space. The proof of optimality of this algorithm (assuming the UGC) uses the invariance principle. }

\subsection{Our contributions}
In this paper, we consider two maximization CSPs, namely, MAX-3-EQUAL and MAX-k-CSP.  Since we are dealing with maximization problems, we set the (usual) convention that a {(randomized)} algorithm is said to give an $\alpha$-approximation (for $\alpha \le 1$) if {(in expectation over the randomness of the algorithm), the value of the output is at least $\alpha$ times the optimal value.}

We first start by describing our result for MAX-3-EQUAL. 
In MAX-3-EQUAL, the variables are boolean-valued and every constraint consists of three literals and it is satisfied if and only if all the three literals are either all zeros or all ones. We show that assuming the Unique Games Conjecture, the MAX-3-EQUAL problem is  $\alpha_{EQU} \approx 0.796$ hard to approximate in polynomial time.  On the complementary side, we also provide a polynomial time algorithm for this problem with the approximation ratio $\alpha_{EQU}$. More formally, we prove  

\begin{theorem} \label{thm:ours}
There is a polynomial time approximation algorithm for the MAX-3-EQUAL problem which achieves the following approximation ratio: 
$$
{\alpha_{EQU}} := \inf_{\delta \in (0,1]}  \frac{ 1 -   \frac{3\cos^{-1} (1-\delta)}{2\pi}}{1-\frac{3\delta}{4}} \approx 0.796  .
$$
Assuming the Unique Games Conjecture, for every $\delta > 0$ 
there is no polynomial time that provides a better 
approximation ratio than ${\alpha_{EQU}}+\delta$. 
\end{theorem}

The hardness proof uses a recent Gaussian noise stability result of Isaksson and Mossel~\cite{IM12} which does not seem to have been previously used in the literature for proving hardness of approximation results.
 In fact, all previous optimal hardness of approximation results with a ``non-trivial" approximation ratio were dependent on the Gaussian noise stability result of Borell~\cite{Borell:85} eg.~MAX-CUT, MAX-2-SAT. 

We also give an analytic proof of the performance of the random hyperplane rounding algorithm on the generic SDP for MAX-3-EQUAL (from~\cite{Rag:08}) showing that the approximation ratio achieved by this rounding algorithm is exactly $\alpha_{EQU}$.\footnote{We actually do a variant of the random hyperplane rounding algorithm where we sample normal random variables with the covariance matrix given by the SDP vectors. Then each variable is assigned $0$ or $1$ depending on the sign of the corresponding normal random variable. Our analysis goes through even if the actual random hyperplane algorithm is used.} Our proof is computer assisted but completely rigorous.  We note that while Zwick~\cite{Zwick98} also considers this problem and analyzes the performance of this algorithm, the analysis is a computer based search and he notes that there is a possibility of the search having missed the worst instance for the rounding algorithm. Nevertheless, the claimed optimum in~\cite{Zwick98} is same as the optimum of our SDP\footnote{{We elaborate on the difference between Zwick's SDP and our SDP in Section~\ref{sec:rounding}.}}. 

\newa{
\begin{remark} 
After the publication of the preprint, David Williamson~\cite{Will:13} informed us that our analysis of the SDP is essentially identical to the analysis of MAX-DICUT SDP from \cite{GW95}. Thus, the analysis from \cite{GW95} can be plugged in to give a much shorter proof for the performance of our algorithm.  
\end{remark}
}


While revisiting the study of the relationship between Gaussian partitions and UGC hardness, we additionally prove hardness results for MAX-k-CSPs. In particular, we investigate the hardness of the MAX-k-AND predicate,  i.e., every constraint consists of $k$ literals $\ell_1, \ldots, \ell_k$ and the constraint is satisfied if and only if $\ell_1 = \ldots = \ell_k=1$.
Following~\cite{Mossel:10} and~\cite{AM:09} by using the fact that in Gaussian space, pair-wise independence implies independence, we prove the following theorem:
\begin{theorem}\label{thm:maxand}
Assuming the Unique Games Conjecture, for every $\eta>0$, there is no polynomial time approximation algorithm that provides an approximation ratio better than $\frac{\lceil (k+1)/2 \rceil}{2^{k-1}}$ for the MAX-k-AND problem. 
\end{theorem}
 This improves upon  \cite{AM:09} where it was shown that MAX-k-CSP is $(k +O(k^{0.525}))/2^k$ hard to approximate. Assuming the Hadamard Conjecture, they could improve it to $\lceil (k+1) /4 \rceil/2^{k-2}$. 

It is worth mentioning  that   \cite{AM:09} proves the aforementioned hardness for a very general class of predicates (ones whose satisfying assignments support pairwise independent distributions) but MAX-k-AND is not included in that class of CSPs. Another important point of difference is that \cite{AM:09} shows that given a MAX-k-CSP with optimal value $1-\eta$, it is (Unique Games) hard to find an assignment which satisfies $\frac{k +O(k^{0.525})}{2^k} + \eta$ fraction of the constraints (for any $\eta>0$). In terms of PCPs, the PCP in \cite{AM:09} has near perfect completeness. This in fact is true even for an earlier paper on hardness of MAX-k-CSPs by Samorodnitsky and Trevisan~\cite{ST06}. In contrast, our result shows that given an instance of  MAX-k-CSP with optimal value $\frac{1}{2\lceil (k+1)/2 \rceil} -\eta$, it is hard to find an assignment satisfying more than $\frac{1}{2^k} +\eta$ fraction of the constraints. 

We do remark that while our improvement over \cite{AM:09} might seem very minor, {Makarychev and Makarychev~\cite{MM:12} give a $0.62 k /2^k$} approximation algorithm for MAX-k-CSP over boolean alphabet. This shows that in some sense, the scope of improvement in the existing hardness results for MAX-k-CSPs is rather limited. Of course, the question of closing the gap between our hardness result and the performance of the algorithm of Charikar \etal remains open. 

\textbf{Overview of proofs of hardness:} The two main novelties in our paper are:
\begin{itemize}
\item Use of the new Gaussian stability result of Isaksson and Mossel~\cite{IM12} to construct a ``dictatorship" test for MAX-3-EQUAL.
\item Use of the ``obvious" Gaussian stability result (i.e., stable partitions for independent gaussians) in a new context to construct a ``dictatorship" test for MAX-k-AND. 
\end{itemize}
In particular, both these dictatorship tests are constructed by a careful combination of a ``good" choice of distribution (for the dictatorship test) and the relevant Gaussian stability result (along with the Invariance principle). Given the dictatorship test, getting the corresponding Unique Games hardness result is rather standard (see \cite{KKMO07,Rag:08}). For the sake of completeness, we give a complete proof of for hardness of MAX-3-EQUAL using the corresponding dictatorship test. For MAX-k-AND, we do not show the conversion of the dictatorship test to a Unique Games hardness result  as the proof is completely analogous to that of MAX-3-EQUAL.

To show the tightness of the UG-hardness result for MAX-3-EQUAL, we also devote a major part of the paper towards analyzing the  performance of our rounding algorithm on the generic SDP from \cite{Rag:08} and showing that it indeed matches the  hardness result. {We would  like to emphasize that while the Gaussian stability result of \cite{IM12} applies to a set of $k$ Gaussian variables (for any $k$), we do not know if this can yield a tight hardness result for MAX-k-EQUAL. In particular, while the Gaussian stability result will imply some hardness of approximation for MAX-k-EQUAL, currently, we do not have an algorithm whose approximation ratio provably matches the hardness result. We elaborate more on this in Section~\ref{sec:difficult}. }


%

 \subsection{Organization} Section~\ref{sec:fourier} states all the fourier analytic and other technical preliminaries required for this paper. 
 Section~\ref{sec:dic} describes a dictatorship test where the tester checks for equality of three literals. Section~\ref{sec:and} describes a dictatorship test where the tester checks if all the $k$ literals are $1$. Section~\ref{sec:UGC} has the two main theorems of this paper, namely a UG-hardness result for the MAX-3-EQUAL problem and a UG-hardness result for MAX-k-AND.   Section~\ref{sec:rounding} describes a SDP relaxation and a rounding algorithm for the MAX-3-EQUAL problem showing the tightness of the hardness result.
 \section{Preliminaries}\label{sec:fourier}

\subsection{Basics of Fourier analysis}

Our proofs are significantly dependent on  fourier analysis.  We start by giving several important definitions. For a more extensive reference, see lecture notes by Mossel~\cite{Mossel:05a}.

We recall that any function $f : \{-1,1\}^n \to \R$ can be written as a multi-linear polynomial. 
\[
f(x) = \sum_{S \subset [n]} \hat{f}(S) x_S, 
\]
where $x_S = \prod_{i \in S} x_i$. Moreover, considering the uniform measure over $\{-1,1\}^n$, we have: 
\[
\E[f] = \hat{f}(\emptyset), \quad Var[f] = \sum_{S \neq \emptyset} \hat{f}^2(S).
\]
The $i$'th influence of $f$ is given by 
\[
I_i(f) := {\mathop{\mathbf{E}}_{x_1,\ldots,x_{i-1},x_{i+1},\ldots,x_n} [Var[ f | x_1,\ldots,x_{i-1},x_{i+1},\ldots,x_n]] = 
\sum_{S : i \in S} \hat{f}^2(S).}
\]

\subsection{Noise operators and their properties}
We will also require the notion of noise operators. 
We consider a particularly important instantiation of the Bonami-Beckner operator namely that on functions over the boolean hypercube $\{-1,1\}^n$ equipped with the uniform measure. 
\begin{definition}
For $\rho \in [-1,1]$, we define the Bonami-Beckner operator $T_{\rho}$ on functions $f: \{-1,1\}^n \rightarrow \mathbb{R}$ as follows. 
$$
T_{\rho} f (x) = \mathop{\mathbf{E}}_{y \sim_{\rho} x} [ f(y) ]  , 
$$
where each coordinate $y_i$ is set to be $x_i$ independently with probability $(1+\rho)/2$ and $-x_i$ with probability $(1-\rho)/2$.  
\end{definition}
The effect of the Bonami-Beckner operator $T_{\rho}$ can be conveniently expressed in terms of the fourier spectrum of a function. In particular, if $f$ is as above, then$$
T_{\rho} f (x) = \sum_{S \subseteq [n]} \hat{f}(S) \rho^{|S|} \chi_S(x) . 
$$

The following standard lemma proves a bound on the number of coordinates with high influence on a function after applying the Bonami-Beckner operator on it, see e.g. \cite{KKMO07}.
\begin{lemma}\label{lem:upbound}
Let $f : \{-1,1\}^n \rightarrow[0,1]$ and $\tau,\gamma>0$. If $\mathcal{A}(f)  = \{ i : \Inf_i (T_{1-\gamma}f) \ge \tau \}$, then $|\mathcal{A}(f)| \le 1/(\gamma \tau) $.
\end{lemma}

The next lemma is a specialization of Lemma~6.2 from \cite{Mossel:10}. It says that expected value of product of polynomials does not change by a lot when noise is added provided individual coordinates come from correlated probability spaces such that no coordinate is absolutely fixed given rest of the coordinates. 
\begin{lemma}\label{lem:mos}
For $1 \leq i \leq n$, 
let $(\Omega_i,\mu_i) = (\{-1,1\}^k,\mu_i)$ where 
\[
\min_{x \in \{-1,1\}^k} \mu_i(x) \geq \alpha >0.
\]
Let $(\Omega,\mu) = \prod_{i=1}^n (\Omega_i,\mu_i) $. 
For $1 \leq a \leq k$, let $\mu_i^a$ be the $a$'th marginal of 
$\mu_i$, in other words 
\[
\mu_i^a(x) = \mu_i(\{(x_1,\ldots,x_k) : x_a = x\}).
\]
Let $\mu^a = \prod_{i=1}^n \mu_i^a$. 
An element $x \in \Omega$ is a $k \times n$ matrix. 
We write $x^a$ for the $a$'th row of $x$ which is distributed 
according to $\mu^a$. 
For $1 \leq a \leq k$, let 
$Q_a$ be a multilinear polynomials $Q_a : \{-1,1\}^n \to [-1,1]$.  
Then, for all $\epsilon >0$, $\exists \gamma = \gamma (\epsilon,\alpha)>0$ such that 
$$
{\left| \mathbf{E} \left[ \prod_{a=1}^k Q_a(x^a) \right] -  \mathbf{E} \left[\prod_{a=1}^k {T_{1-\gamma}} Q_a(x^a) \right] \right| \le \epsilon k  . }
$$
\end{lemma}
\subsection{Gaussian Stability results}
The following theorem from Isaksson and Mossel~\cite{IM12} is the main technical result that we use here.
 
\begin{theorem}\label{th:IM} {\emph{(Theorem~5.1, \cite{IM12})}}
Let $\Omega = \{-1,1\}^k$, $\rho \in [0,1]$ and let $\mu$ be a probability distribtion over $\Omega$ such that 
\begin{itemize}
\item 
$\mu(x) \geq \alpha > 0$ for all $x$. 
\item 
For $s,t \in \{-1,1\}$ and all $1 \leq a \neq b \leq k$: 
\[
\mu(x_a = s,  x_b = t) = 
\frac{1}{2} \rho \delta(s,t) + \frac{1}{4} (1-\rho),
\]
\end{itemize}
where $\delta(s,t)=1$ iff $s=t$. 
Consider the space $(\Omega^n,\mu^n)$. 
An element $x \in \Omega^n$ may be viewed as a $k \times n$ 
matrix. Write $x^a$ for the $a$'th row of this matrix for 
$1 \leq a \leq k$. Note that $x^a$ is uniformly distributed in 
$\{-1,1\}^n$. 

Then for every $\epsilon>0$, $\exists \tau = \tau(\epsilon,k,\alpha)>0$ such that for any $f_1, \ldots, f_k : \{-1,1\}^n \rightarrow [0,1]$ satisfying $\max_{i,j} \Inf_i(f_j) \le \tau$, 
$$
\mathbf{E} \left[ \prod_{a=1}^k f_a(x^a) \right] \le \Pr[\forall a \in [k] : \mathcal{Z}_a \le t_j] + \epsilon,
$$
where $\mathcal{Z}_1, \ldots, \mathcal{Z}_k \sim \mathcal{N}(0,1)$ are jointly normal and $\Cov(\mathcal{Z}_a, \mathcal{Z}_{a'}) = \rho$ for all $a \ne a'$ and each $t_j$ is chosen so that $\Pr[\mathcal{Z}_a \le t_a] = \mathbf{E} [f_a]$. 
\end{theorem}
To intuitively understand the above theorem, consider the case when $f_1= \ldots =  f_k=f$ has range $\{0,1\}$. Also, let $x^1,  \ldots, x^k \in \{-1,1\}^n$ such that each $x^i$ is uniform in $\{-1,1\}^n$ and for any $j \in [n]$ and $i  \not = \ell \in [k]$, the $j^{th}$ bit of $x^i$ and $x^\ell$ are $\rho$-correlated. Let us equip $\mathbb{R}^n$ with the standard normal measure and  define the function $\tilde{f} : \mathbb{R}^n \rightarrow \{-1,1\}$ as follows : $\tilde{f} : x \mapsto sgn (x_1 -\theta)$ where $x_1$ is the first coordinate of $x$ and $\theta$ is chosen so that $\mathbf{E}_{x \in \{-1,1\}^n} [f(x)] = \mathbf{E}_{x \in \mathcal{N}^n(0,1)} [\tilde{f}(x)]$. Then, for all ``low-influence" function $f$, the probability that $\forall \ell \in [k]$, $f(x^{\ell}) =1$ is upper bounded by the probability that $\forall \ell \in [k]$, $\tilde{f}(x^{\ell}) =1$. 
We also consider the corollary of the above theorem when $\rho=0$. We do remark that the following corollary can actually be obtained using the Invariance principle from Mossel~\cite{Mossel:10} and does not require the full strength of \cite{IM12}. 
\begin{corollary}\label{Cor:gauss}
Let $\Omega = \{-1,1\}^k$ and let $\mu$ be a probability distribtion over $\Omega$ such that 
\begin{itemize}
\item 
$\mu(x) \geq \alpha > 0$ for all $x$. 
\item 
For $s,t \in \{-1,1\}$ and all $1 \leq a \neq b \leq k$: 
\[
\mu(x_a = s,  x_b = t) = \frac{1}{4}. 
\]
\end{itemize}

Consider the space $(\Omega^n,\mu^n)$. 
An element $x \in \Omega^n$ may be viewed as a $k \times n$ 
matrix. Write $x^a$ for the $a$'th row of this matrix for 
$1 \leq a \leq k$. Note that $x^a$ is uniformly distributed in 
$\{-1,1\}^n$. 

Then for every $\epsilon>0$, $\exists \tau = \tau(\epsilon,k,\alpha)>0$ such that for any $f_1, \ldots, f_k : \{-1,1\}^n \rightarrow [0,1]$ satisfying $\max_{i,j} \Inf_i(f_j) \le \tau$, 
$$
\mathbf{E} \left[\prod_{a=1}^k f_a(x^a) \right]\le \prod_{a=1}^k 
\E[f_a] + \epsilon.
$$

\end{corollary}
\begin{proof}
The corollary follows by putting $\rho=0$ in Theorem~\ref{th:IM} and then observing that $\mathcal{Z}_1, \ldots, \mathcal{Z}_k \sim \mathcal{N}(0,1)$ in the conclusion of Theorem~\ref{th:IM} are simply i.i.d. $\mathcal{N}(0,1)$ random variables. 
\end{proof}

\subsection{Useful facts}
We will require the following very useful fact about Gaussians. 
 For a reference, see~\cite{Bacon:75}. 
\begin{fact}\label{fac:ac}
Let $\mathcal{X}, \mathcal{Y}, \mathcal{Z} \sim \mathcal{N}(0,1)$ such that $(\mathcal{X}, \mathcal{Y}, \mathcal{Z})$ are jointly normal and $\Cov(\mathcal{X},\mathcal{Y}) =\rho_1$, $ \Cov(\mathcal{Z},\mathcal{Y}) = \rho_2$ and $\Cov(\mathcal{X},\mathcal{Z}) =\rho_3$.  Then, 
$$\Pr[X,Y,Z \le 0] = \Pr[X,Y,Z \ge 0]=\frac12 -  \frac{\cos^{-1} \rho_1 + \cos^{-1} \rho_2 + \cos^{-1} \rho_3}{4\pi}.$$
\end{fact}

We will also use the following very useful construction of pairwise independent distribution (which can be found in \cite{BGP12,Prebor:89}). 
\begin{fact}\label{fac:distro}
For any $k \in \mathbb{N}$, there is a distribution $D_k$ on $\{-1,1\}^k$ such that the following holds: 
\begin{itemize}
\item For any $i \in [k]$, $\mathbf{E} [x_i]=0$.
\item For any $i, j \in [k]$ and $i \not  = j$, $\mathbf{E} [x_i x_j] =0$, i.e., any two coordinates are pairwise independent. 
\item $\Pr_{x \in D_k} [x_1 = \ldots =x_k =1] = \frac{1}{2\lceil (k+1)/2 \rceil}$.
\end{itemize}
\end{fact}
\begin{proof}
We will construct a symmetric distribution $D_k$ with the above mentioned properties. First, we consider the case when $k$ is odd. In this case, define $D_k$ as follows: 
$$
D_k(x) = \begin{cases} \frac{1}{k+1} &\mbox{if } x =(1,\ldots, 1), \\ 
\frac{k}{k+1} \cdot \frac{1}{\binom{k}{(k+1)/2}}& \mbox{if } \littlesum_{i=1}^k x_i=-1, \\
0& \mbox{otherwise.} \end{cases} 
$$
It is easy to verify that all the three required properties hold for this construction of $D_k$. We next move to the case when $k$ is even. In this case, we define $D_k$ as
$$
D_k(x) = \begin{cases} \frac{1}{k+2} &\mbox{if } x =(1,\ldots, 1), \\ 
\frac{1}{2} \cdot \frac{1}{\binom{k}{k/2}}& \mbox{if } \littlesum_{i=1}^k x_i=0, \\
\frac{k}{2k+4} \cdot \frac{1}{\binom{k}{1 + \frac{k}{2}}}& \mbox{if } \littlesum_{i=1}^k x_i=-2, \\ 0& \mbox{otherwise.} \end{cases} 
$$
Again, it is easy to verify that all the three properties required of $D_k$ hold for this construction. 
\end{proof}

\section{Dictatorship test for MAX-3-EQUAL}\label{sec:dic}
In this section, we will construct a dictatorship test where the tester checks for equality of 3 literals. More precisely, we will prove the following theorem: 
\begin{theorem}\label{thm:dict}
For any $0 < \delta <1$ and $\epsilon>0$, there is a distribution $D_{\delta}^n$ over $(\{-1,1\}^{n})^3$ such that if $(X,Y,Z) \sim D_{\delta}^n$, then for every $f : \{-1,1\}^n \rightarrow [0,1]$ with $\mathbf{E} [f] = 1/2$,
\begin{itemize}
\item If $f(x) = (1+ x_i)/2$ for some $i \in [n]$, then $$\mathop{\mathbf{E}}_{(X,Y,Z) \sim D_{\delta}^n} [ f(X) \cdot f(Y) \cdot f(Z) + (1- f(X)) \cdot (1- f(Y)) \cdot  (1-f(Z)) ] = 1- \frac{3\delta}{4} .$$
\item  $\exists \tau = \tau(\delta, \epsilon)>0$ and $\eta = \eta(\delta,\epsilon)>0$ such that if $\max_{i} \Inf_i (T_{1-\eta} f) \le \tau$,
$$\mathop{\mathbf{E}}_{(X,Y,Z) \sim D_{\delta}^n} [ f(X) \cdot f(Y) \cdot f(Z) + (1- f(X)) \cdot (1- f(Y)) \cdot  (1-f(Z)) ] \le 1 - (3 \cos^{-1} (1-\delta))/2\pi  + \epsilon. $$
\end{itemize} 
\end{theorem}
Before starting the proof, we note that if $f$ were a boolean function with range $\{0,1\}$, then $f(X) \cdot f(Y) \cdot f(Z) + (1- f(X)) \cdot (1- f(Y)) \cdot  (1-f(Z))$ is $1$ if and only if $f(X)  = f(Y) = f(Z)$. Thus, we have a dictatorship test which checks for equality of $3$ bits. 
\begin{proof}
Let us define a distribution $D_{\delta}$  over $\{-1,1\}^3$ as follows: 
$$
D_{\delta}(x) =  \begin{cases}  \frac12 - \frac{3 \delta}{8} &\mbox{if }x=(1,1,1) \mbox{ or }  x=(-1,-1,-1), \\
\frac{\delta}{8} &\mbox{ otherwise.} \end{cases}
$$
 Let $D_1, D_2, \ldots, D_n$ be $n$ i.i.d. samples of $D_{\delta}$. Let $D_i(j)$ denote the $j^{th}$ bit of $D_i$. With this, let us define $X,Y, Z \in \{-1,1\}^n$ as  $$X= (D_1(1), \ldots, D_n(1)) \quad Y= (D_1(2), \ldots, D_n(2))\quad Z = (D_1(3), \ldots, D_n(3)).$$ We let the joint distribution $(X,Y,Z)$ as defined here be $D_{\delta}^n$. We start with the proof of the first item. 
\newline
Completeness: Note for any particular $i \in [n]$, the $i^{th}$ coordinate of $D_{\delta}$ has the same string with probability $1-3 \delta/4$. Now, if $f(x) = (1 + x_i)/2$, then  it means that $f(x) =1$ if $x_i=1$ and $0$ otherwise. Hence, we have 
\begin{eqnarray*}
\mathbf{E}_{(X,Y,Z) \in D_{\delta}^n} [f(X) \cdot f(Y) \cdot f(Z) + (1- f(X)) \cdot (1- f(Y)) \cdot  (1-f(Z)) \\ = \mathbf{E}_{(X,Y,Z) \in D_{\delta}^n} [\mathbf{I}(X_i = Y_i = Z_i)] = 1-3\delta/4,
\end{eqnarray*}
where $\mathbf{I}(P)$ denotes the indicator function for the predicate $P$. This finishes the proof of the first item. We next do the proof of the second item. 
\newline
Soundness: Let $\mathcal{Q}$ be the multilinear polynomial representation of $f$. Note that for any $x \in \{-1,1\}^n$, $|\mathcal{Q}(x)| \le 1$.  Let $\Omega$ be the probability space with domain $\{-1,1\}^3$ and probability measure $D_{\delta}$ on it. Note that $\forall x \in \{-1,1\}^3$,   $D_{\delta}(x)  \ge \delta/8$.  Hence, by Lemma~\ref{lem:mos}, we get that $\exists \eta = \eta(\delta, \epsilon)>0$, {such that}, 
\begin{equation}\label{eq:22}
| \mathbf{E}_{} [f(X) \cdot f(Y) \cdot f(Z) - T_{1-\eta} f(X) \cdot T_{1-\eta}f(Y) \cdot T_{1-\eta}f(Z)] | \le \epsilon/4.
\end{equation}
 Likewise, we get that 
\begin{equation}\label{eq:21}
| \mathbf{E}_{} [(1-f(X)) \cdot (1-f(Y)) \cdot (1-f(Z)) - (1-T_{1-\eta} f(X)) \cdot (1- T_{1-\eta}f(Y)) \cdot (1-T_{1-\eta}f(Z))] | \le \epsilon/4.
\end{equation}
In the last two equations, $(X,Y,Z) \sim D_{\delta}^n$. We now apply Theorem~\ref{th:IM}. In particular, note that if $(X,Y,Z) \sim D_{\delta}^n$, then the variables $(X_i,Y_i,Z_i)$ are independently and identically distributed.  Also, for any $i \in [n]$, $X_i$, $Y_i$ and $Z_i$ are pairwise $\rho = (1-\delta)$ correlated and for any $(x,y,z) \in \{-1,1\}^3$, $\Pr[(X_i, Y_i,Z_i)=(x,y,z)] \ge \delta/8>0$.  Finally, note that $X$, $Y$ and $Z$ are distributed as $U_n$. Hence $$
\mathbf{E}_{X} [f(X)] =\mathbf{E}_{Y} [f(Y)] = \mathbf{E}_{Z} [f(Z)] = 1/2.$$ 
As the Bonami Beckner operator preserves expectation of the function under the uniform distribution, we get 
$$
\mathbf{E}_{X} [T_{1-\eta} f(X)] =\mathbf{E}_{Y} [T_{1-\eta} f(Y)] = \mathbf{E}_{Z} [T_{1-\eta}f(Z)] = 1/2.
$$
Thus, by Theorem~\ref{th:IM}, $\exists \tau = \tau(\delta, \epsilon)$ such that if $\max_i \Inf_i (T_{1-\eta} f) \le \tau$, then we have
$$
|\mathbf{E}_{(X,Y,Z) \in D_{\delta}^n} [T_{1-\eta} f(X) \cdot T_{1-\eta}f(Y) \cdot T_{1-\eta}f(Z)] | \le \Pr[\mathcal{X}, \mathcal{Y}, \mathcal{Z} \le 0] + \epsilon/4,
$$
where $\mathcal{X}, \mathcal{Y}, \mathcal{Z} \sim \mathcal{N}(0,1)$ and $\Cov(\mathcal{X},\mathcal{Y}) = \Cov(\mathcal{Z},\mathcal{Y}) = \Cov(\mathcal{X},\mathcal{Z}) = 1-\delta$.  
Here, we again assume that  $\tau$ in the hypothesis of the theorem is  sufficiently small so that the hypothesis of Theorem~\ref{th:IM} is valid. 
Likewise, we get that 
$$
|\mathbf{E}_{(X,Y,Z) \in D_{\delta}^n} [(1-T_{1-\eta} f(X)) \cdot (1-T_{1-\eta}f(Y)) \cdot (1-T_{1-\eta}f(Z))] | \le \Pr[\mathcal{X}, \mathcal{Y}, \mathcal{Z} \le 0] + \epsilon/4.
$$
Combining the above with (\ref{eq:21}) and (\ref{eq:22}), we get that 
$$
\mathbf{E}_{(X,Y,Z) \in D_{\delta}^n} [f(X) \cdot f(Y) \cdot f(Z) + (1- f(X)) \cdot (1- f(Y)) \cdot  (1-f(Z)) ] \le 2\Pr[\mathcal{X}, \mathcal{Y}, \mathcal{Z} \le 0] + \epsilon.
$$
Using Fact~\ref{fac:ac}, we conclude that 
$$
\mathbf{E}_{(X,Y,Z) \in D_{\delta}^n} [f(X) \cdot f(Y) \cdot f(Z) + (1- f(X)) \cdot (1- f(Y)) \cdot  (1-f(Z)) ] \le 1 -   \frac{3\cos^{-1} (1-\delta)}{2\pi}+ \epsilon.
$$
completing the proof. 
\end{proof}

\section{Dictatorship test for MAX-k-AND}\label{sec:and}
 In this section, we construct a dictatorship test for MAX-k-AND i.~e.~the tester checks if a particular set of $k$ literals are all set to $1$.  For the purposes of this section, let us assume $\rho(k) = \frac{1}{2\lceil (k+1)/2 \rceil}$. \begin{theorem}\label{thm:dictand}
For any $k\ge 3$ and $\delta>0$, there is a distribution $D$ over $(\{-1,1\}^{n})^k$ such that if $(X_1, \ldots, X_k) \sim D$ such that for every $f : \{-1,1\}^n \rightarrow [0,1]$ with $\mathbf{E} [f] = 1/2$,
\begin{itemize}
\item If $f(x) = (1+ x_i)/2$ for some $i \in [n]$, then $$\Pr_{(X_1, \ldots, X_k) \sim D} [ f(X_1) \cdot \ldots \cdot f(X_k)] \ge \rho(k) - \delta.$$
\item  $\exists \tau = \tau(\delta,k)>0$ and $\eta = \eta(\delta,k)>0$ such that if $\max_{i} \Inf_i (T_{1-\eta} f) \le \tau$,
$$\Pr_{(X_1, \ldots, X_k) \sim D} [ f(X_1) \cdot \ldots \cdot f(X_k)] \le \frac{1}{2^k} + \delta.  $$
\end{itemize} 
\end{theorem}
We remark that
if   $f$ were to take values in $\{0,1\}$, then we note that $f(X_1) \cdot \ldots \cdot f(X_k)=1$ if and only if $f(X_1) \wedge \ldots \wedge f(X_k)=1$. 
\begin{proof}
Let $D_k$ be the distribution from Fact~\ref{fac:distro}.  We let $\xi= \delta/4$.  Now, we let $D_{\xi} = (1 - \xi) D_k + \xi U_k$.  Let $D_1, \ldots, D_n$ be  $n$ i.i.d. samples from $D_{\xi}$. Let $D_i(j)$ be the $j^{th}$ bit of $D_i$. Having done this, we define  $X_j$ for $1 \le j \le k$ as $X_j = (D_1(j),\ldots, D_n(j))$.  
Let $D$ be defined as the joint distribution of $(X_1, \ldots, X_k)$. 

As before, we start with the proof of the first item. 
\newline
Completeness: Since $f(x) = (1+x_i)/2$ (for some $i \in [n]$),  it means that $f(x) =1$ if $x_i=1$ and $0$ otherwise. Hence, we have 
\begin{eqnarray*}
\mathbf{E}_{(X_1, \ldots, X_k) \in D} [f(X_1) \cdot \ldots  \cdot f(X_k)]   &=& \mathbf{E}_{(X_1, \ldots, X_k) \in D} [\mathbf{I}(X_1(i) = \ldots = X_k(i)=1)] \\ &=& \rho(k) (1- \xi) + \xi 2^{-k} \ge \rho(k) -\delta.
\end{eqnarray*}
\newline
Soundness: Let $\mathcal{Q}$ be the multilinear polynomial representation of $f$. Note that for any $x \in \{-1,1\}^n$, $|\mathcal{Q}(x)| \le 1$.  Let $\Omega$ be the probability space with domain $\{-1,1\}^k$ and probability measure $D_{\xi}$ on it. Observe that $D_{\xi}(x)  \ge \xi \cdot 2^{-k}$ for all $x \in \{-1,1\}^k$.  Hence, by Lemma~\ref{lem:mos}, we get that $\exists \eta = \eta(\xi,k)>0$ (note because $\xi=\delta/4$, we can also express $\eta$ as a function of $\delta$ and $k$ as required by the theorem), 
\begin{equation}\label{eq:22a}
| \mathbf{E}_{(X_1,\ldots,X_k) \in D} [f(X_1) \cdot \ldots \cdot f(X_k) - T_{1-\eta} f(X_1) \cdot \ldots \cdot T_{1-\eta}f(X_k)] | \le \frac \xi 4.
\end{equation}
We can now apply Corollary~\ref{Cor:gauss} to the function $T_{1-\eta} f$ and the random variables $(X_1, \ldots, X_k) \sim D$. Much like in the proof of Theorem~\ref{thm:dict}, it is easy to check that all the conditions are satisfied (In particular, note that for any $i \in [n]$, $X_1(i), X_2(i), \ldots, X_k(i)$ are pairwise independent). 
By Corollary~\ref{Cor:gauss}, $\exists \tau= \tau(\xi,k)$ such that if $\max_i Inf_i(f) \le \tau$, we have
\begin{equation}\label{eq:random}
|\mathbf{E}_{(X_1,\ldots,X_k) \in D} [T_{1-\eta} f(X_1) \cdot \ldots \cdot T_{1-\eta}f(X_k)] | \le 2^{-k} + \frac{\xi}{4}.
\end{equation}
As before, we note that $\tau(\xi,k)$ can be expressed as $\tau(\delta,k)$. Here, we are assuming that the $\eta(\xi,k)$ and $\tau(\xi,k)$ chosen to be sufficiently small so that the hypothesis of Corollary~\ref{Cor:gauss} is valid. 
Combining (\ref{eq:22a}) and (\ref{eq:random}), we get that 
$$
|\mathbf{E}_{(X_1,\ldots,X_k) \in D} [ f(X_1) \cdot \ldots \cdot f(X_k)] | \le 2^{-k}  + \frac \xi 2.
$$
\end{proof}

\section{Unique games hardness from Dictatorship test}\label{sec:UGC}
In this section, we use the dictatorship tests constructed in Section~\ref{sec:dic} and Section~\ref{sec:and} to show the following theorems. 
\begin{theorem}\label{thm:ugceq}
Assuming the Unique Games Conjecture, for every $0<\delta<1$ and $\epsilon>0$, it is NP-hard to distinguish an instance of MAX-3-EQUAL with value $1-3\delta/4 - \epsilon$ from an instance of value $ 1 -  \frac{3\cos^{-1} (1-\delta)}{2\pi}+ \epsilon$.   In other words, for every $\epsilon > 0$, MAX-3-EQUAL is $\alpha_{EQU} + \epsilon$ hard to approximate where 
$$
\alpha_{EQU} = \inf_{\delta \in (0,1)}  \frac{1 -  \frac{3\cos^{-1} (1-\delta)}{2\pi}}{1-\frac{3\delta}{4}} \approx 0.796.
$$
\end{theorem}
\begin{theorem}\label{thm:ugcand}
Assuming the Unique Games Conjecture, for every  $\epsilon>0$, it is NP-hard to distinguish an instance of MAX-k-AND with value $\frac{1}{2\lceil (k+1)/2 \rceil} -\epsilon$ from an instance of value $2^{-k}+ \epsilon$.   In other words, for every $\epsilon > 0$, MAX-k-AND is $\frac{\lceil (k+1)/2 \rceil}{2^{k-1}} + \epsilon$ hard to approximate. \end{theorem}

Theorem~\ref{thm:ugceq} uses the dictatorship test in Theorem~\ref{thm:dict} to reduce Unique Games to MAX-3-EQUAL. Similarly, Theorem~\ref{thm:ugcand} uses the dictatorship test in Theorem~\ref{thm:dictand} to reduce Unique Games to MAX-k-AND. As we said in the introduction, these reductions are by now very standard and can be found  {in} several places. For the sake of convenience of the reader, we include the full proof of Theorem~\ref{thm:ugceq}. The proof of Theorem~\ref{thm:ugcand} is exactly analogous and hence, we do not do it here. 

We begin by defining the Unique Label Cover problem {and then state} Khot's Unique Games Conjecture (slightly differently stated than Conjecture~\ref{conj:1}).
\begin{definition}
{An instance of a  Unique Label Cover problem} $(G,\Sigma)$ on alphabet size $t$ is defined by a graph $G = (V,E)$ and a set of permutations $\Sigma = \{ \sigma_{(u,v)} : [t] \rightarrow [t] \}_{(u,v) \in E}$. For any map $\mathcal{L} : V \rightarrow [t]$ and $(u,v) \in E$, $\mathcal{A}_{{\mathcal{L}}}(u,v) = 1$ if and only if $\mathcal{L}(v) = \sigma_{(u,v)} ( \mathcal{L}(u))$, otherwise it is zero. For a map $\mathcal{L} : V \rightarrow [t]$, $val_{\mathcal{L}}(G) = \mathbf{E}_{(u,v) \sim E} [\mathcal{A}_{\mathcal{L}}(u,v)]$. {The value of the unique label cover instance is }(denoted by) $val(G)  = \max_{\mathcal{L} : V \rightarrow [t]} val_{\mathcal{L}}(G) $. 
\end{definition}
\begin{conjecture}\cite{Kho:02} \textbf{Unique Games Conjecture: }
For every $\epsilon>0$, there is a $t = t(\epsilon)$ such given a unique label cover problem $(G,\Sigma)$ on alphabet size $t$, distinguishing whether $val(G) \le \epsilon$ or $val(G) \ge 1-\epsilon$ is NP-hard.  We can also assume that the graph $G$ is regular. 
\end{conjecture}
Having stated the unique games conjecture, we describe a PCP verifier for the unique label cover problem which checks for equality of $3$ bits. By the standard reduction between PCP verifiers and hardness of approximation, we get a hardness result for the MAX-3-EQUAL problem.

\textbf{Description of the PCP verifier:} Given the unique games instance $(G,\Sigma)$ (on alphabet size $t$), we assume that $V = [n]$ and build a PCP verifier over $n \cdot 2^t$ boolean variables as follows: For every $i \in [n]$, we have a function $f_i : \{-1,1\}^t \rightarrow  \{0,1\}$. Note that any such truth table can be described by $2^t$ boolean variables and hence the family of functions $\{ f_i \}$ can be described in all by $n \cdot 2^t$ variables.  
{\begin{remark}
We will also assume the functions are folded, i.e., for any $x$, $f(x) \not = f(-x)$. Note that this can be done without loss of generality, because whenever the verifier needs to query $f(x)$, if $x_1=1$, it queries $f(x)$. Else it queries $f(-x)$ and flips the output. We note that ``flipping" the output can be implemented by introducing negated literals in the resulting CSP. Also, we observe that dictators satisfy this requirement. 
\end{remark}}
For a given $\delta \in (0,1)$, let $D_{\delta}^t$ be the distribution in the hypothesis of Theorem~\ref{thm:dict}. Note that the distribution $D_{\delta}^t$ is over $(\{-1,1\}^t)^3$. {Also, we use $\circ$ to denote composition of functions. In other words, for two functions $g_1$ and $g_2$, $g_1 \circ g_2 (x)$ denotes $g_1(g_2(x))$.} With this, the verifier is as follows: 
\begin{itemize} 
\item Pick $v \in V$ uniformly at random and choose three random neighbors of $v$, say,  $w_1, w_2, w_3$ uniformly at random. 
\item Choose $(X,Y,Z) \sim D_{\delta}^t$ (described above) and accept if and only if 
$$f_{w_1} \circ \sigma_{(w_1,v)} (X)=f_{w_2} \circ \sigma_{(w_2,v)} (Y)=f_{w_3} \circ \sigma_{(w_3,v)} (Z).$$ 
\end{itemize}
We next show the correctness of this verifier. In other words, we prove the following two lemmas. 
\begin{lemma}\label{lem:complete}
If $val(G) \ge 1- \epsilon$, then there is a set of functions $\{f_i : \{-1,1\}^t \rightarrow \{0,1\}\}_{i \in [n]}$ such that the above verifier accepts with probability at least $(1-3\epsilon)(1-3\delta/4)$. 
\end{lemma}
\begin{lemma}\label{lem:coloring}
For any $\epsilon>0$, if the above verifier passes with probability more than  $1 -   \frac{3\cos^{-1} (1-\delta)}{2\pi}+ \epsilon$, then $\exists \mathcal{L} : V \rightarrow [t]$ such that $val_{\mathcal{L}}(G) = \kappa(\epsilon,\delta) >0$. 
\end{lemma}

Since $val_{\mathcal{L}}(G) $ in conclusion of  $\kappa(\epsilon,\delta)$ does not depend on $t$, hence by  combining Lemma~\ref{lem:complete} and Lemma~\ref{lem:coloring} and the standard reduction between PCPs and hardness of CSPs, we prove Theorem~\ref{thm:ugceq}.  The proofs of  Lemma~\ref{lem:complete} and Lemma~\ref{lem:coloring} follow. 
\begin{proof}[of Lemma~\ref{lem:complete}]
Since $val(G) \ge (1-\epsilon)$, $\exists \mathcal{L} : V \rightarrow [t]$ such that $val_{\mathcal{L}}(G) \ge 1- \epsilon$. Let $\mathcal{L}$ be such a labeling of the vertices. We let $f_i : \{-1,1\}^t \rightarrow \{0,1\}$ be the dictator function corresponding to $\mathcal{L}(i)$. In other words, $f_i(x) = (1+x_{\mathcal{L}(i)})/2$.  Now, since  $val_{\mathcal{L}}(G) \ge 1- \epsilon$ and the constraint graph $G$ is regular,  if we choose $v$ uniformly at random and then a uniform random neighbor $w_i$, then $\mathcal{A}_{\mathcal{L}}(v,w_i) =1$ with probability $1-\epsilon$.
By a union bound,   with probability at least $1-3\epsilon$, $\mathcal{A}_{\mathcal{L}}(v,w_1) =\mathcal{A}_{\mathcal{L}}(v,w_2)=\mathcal{A}_{\mathcal{L}}(v,w_3)= 1$. If this is indeed the case, then, 
$$
f_{w_1} \circ \sigma_{(w_1,v)} =f_{w_2} \circ \sigma_{(w_2,v)} =f_{w_3} \circ \sigma_{(w_3,v)}.
$$
Now, applying the first part of Theorem~\ref{thm:dict}, we get that in this case the test accepts with probability $1 - 3 \delta/4$.  Thus, the total probability that the test accepts is at least 
$
(1-3\epsilon)(1-3\delta/4)
$.
\end{proof}
We next move to the more difficult case of soundness. 

\begin{proof}[of Lemma~\ref{lem:coloring}]
{The proof follows the arguments in \cite{KKMO07} very closely.} We first describe the labeling $\mathcal{L}$ and then describe its correctness. Our labeling is a randomized scheme. Let $\eta, \tau>0$ be two parameters which are chosen according to the second part of the hypothesis of Theorem~\ref{thm:dict} for parameters $\epsilon/2$ and $\delta$. 
First, for every $v \in V$, we define $g_v : \{-1,1\}^t \rightarrow [0,1]$ as 
$$
g_v(x) = \mathbf{E}_{(w,v) \in E} [  f_{w} \circ \sigma_{(w,v)}(x)].
$$
Again for every $v \in V$ we define $\mathcal{A}(v) \subseteq V$ as $$\mathcal{A}(v) = \{ i : \Inf_i(T_{1-\eta} f_v) \ge \tau/2 \} \cup\{ i : \Inf_i(T_{1-\eta} g_v) \ge \tau \}. $$The randomized labeling scheme is the following: If the set $\mathcal{A}(v)$ is empty, $\mathcal{L}(v)$ is chosen arbitrarily. Else, it is chosen to be a uniformly random element from the set $\mathcal{A}(v)$. The following proposition gives us the desired result. 
\begin{proposition}\label{clm:coloring}
Over the choice of randomness for choosing $\mathcal{L}$, {$\mathbf{E} [val_{\mathcal{L}}(G)] \ge (\epsilon\eta^2 \tau^3)/64 $. }
\end{proposition}
By fixing the randomness in the above proposition desirably, we get Lemma~\ref{lem:coloring}. So, the proof boils down to proving Proposition~\ref{clm:coloring}.
\begin{proof}
Let $\mathcal{D}$ be a probability distribution over $V^4$ where $(v,w_1, w_2, w_3) \in \mathcal{D}$ is sampled as follows: $v \in V$ is chosen uniformly at random and $w_1$, $w_2$, $w_3$ are chosen to be three random neighbors of $v$. Further, let $(X,Y,Z) \in D_{\delta}^t$. 
Then, the probability of acceptance of the verifier is given by 
\begin{eqnarray*}
&& \mathop{\mathbf{E}} [\mathbf{I}(f_{w_1} \circ \sigma_{(w_1,v)}(X)= f_{w_2} \circ \sigma_{(w_2,v)} (Y) = f_{w_3} \circ \sigma_{(w_3,v)} (Z))] \\
&=&\mathop{\mathbf{E}}  [ f_{w_1} \circ \sigma_{(w_1,v)}(X) \cdot f_{w_2} \circ \sigma_{(w_2,v)} (Y)  \cdot f_{w_3} \circ \sigma_{(w_3,v)} (Z)] \\
&+& \mathop{\mathbf{E}} [(1- f_{w_1} \circ \sigma_{(w_1,v)}(X)) \cdot (1-f_{w_2} \circ \sigma_{(w_2,v)} (Y))  \cdot (1- f_{w_3} \circ \sigma_{(w_3,v)} (Z))] \\
&=& \mathop{\mathbf{E}} [g_v(X) \cdot g_v(Y) \cdot g_v(Z) + (1-g_v(X)) \cdot (1-g_v(Y)) \cdot (1-g_v(Z))].
\end{eqnarray*} Since the verifier accepts with probability at least $1 -   \frac{3\cos^{-1} (1-\delta)}{2\pi}+ \epsilon$, a Markov argument gives that for at least an $\epsilon/2$ fraction of vertices $v \in V$,
$$
\mathop{\mathbf{E}} [g_v(X) \cdot g_v(Y) \cdot g_v(Z) + (1-g_v(X)) \cdot (1-g_v(Y)) \cdot (1-g_v(Z))] \ge 1 -   \frac{3\cos^{-1} (1-\delta)}{2\pi}+ \epsilon/2.
$$
{Denote this subset (of $V$) by $A$.}   Note that by the second part of Theorem~\ref{thm:dict}, for every $v \in A$, $\exists i \in [t]$, such that $\Inf_{i}(T_{1-\eta} g_v) \ge \tau$.   For every $v \in A$, fix an  $i$ which satisfies $\Inf_i (T_{1-\eta} g_v) \ge \tau$.  
\begin{eqnarray*}
\tau \le \Inf_i (T_{1-\eta} g_v) &=& \sum_{S : i \in S} (1-\eta)^{|S|} \widehat{g_v}(S)^2 = \sum_{S : i \in S} (1-\eta)^{|S|} \left(\mathop{\mathbf{E}_{(w,v) \in E}} [\widehat{f_{w} \circ \sigma_{(w,v)}}(S)] \right)^2 \\
&=&  \sum_{S : i \in S} (1-\eta)^{|S|} \left(\mathop{\mathbf{E}_{(w,v) \in E}} [\widehat{f_{w}} (\sigma_{(w,v)}^{-1}(S))] \right)^2. \end{eqnarray*}
Here $\sigma_{(w,v)}^{-1}(S)$ is the pre-image of the set $S$ under the map $\sigma_{(w,v)}$. Now, by Jensen's inequality we get that 
\begin{eqnarray*}\sum_{S : i \in S} (1-\eta)^{|S|} \left(\mathop{\mathbf{E}_{(w,v) \in E}} [\widehat{f_{w}} (\sigma_{(w,v)}^{-1}(S))] \right)^2 &\le& \mathop{\mathbf{E}}_{(w,v) \in E} {\left[\sum_{S:i \in S} (1-\eta)^{|S|} \widehat{f_{w}}^2 \left(\sigma_{(w,v)}^{-1}(S)\right) \right]}\\ 
&=& \mathop{\mathbf{E}}_{(w,v) \in E} \left[\Inf_{\sigma_{(w,v)}^{-1}(i)} (T_{1-\eta} f_w)\right].
\end{eqnarray*}
Using a Markov argument, this implies that  for such a $v \in A$ and $i$ such that $\Inf_i (T_{1-\eta} g_v) \ge \tau$ , at least a $\tau/2$ fraction of neighbors $w$ of $v$ satisfy,  $$\Inf_{\sigma_{(w,v)}^{-1}(i)} (T_{1-\eta} f_w) \ge \tau/2.$$
{We say that such a pair $(v,w)$ of vertices is ``good".}   Using Lemma~\ref{lem:upbound}, it can be easily shown that for every $v \in V$, $|\mathcal{A}(v)| \le 4/(\tau \eta)$.  This means that for every $v \in A$, the randomized scheme $\mathcal{L}$ assigns $\mathcal{L}(v)=i$ such that $\Inf_i(T_{1-\eta} g_v) \ge \tau$ with probability at least $(\eta \tau)/4$. Observe that for any such $v \in A$, at least $\tau/2$ fraction of its neighbors $w$ are such that $(v,w)$ is ``good". 
Note that for any good pair $(v,w)$, if $\Inf_i(T_{1-\eta} g_v) \ge \tau$, then $ \Inf_{\sigma_{(w,v)}^{-1}(i)} (T_{1-\eta} f_w) \ge \tau/2$. This implies that $\sigma_{(w,v)}^{-1}(i) \in \mathcal{A}(w)$.  Thus, with probability at least $(\tau \eta)/4$,  $\mathcal{L}(w) = \sigma_{(w,v)}^{-1}(i)$.  Thus, overall the probability that $\mathcal{L}(v,w) =1$ is at least {$\epsilon \tau^3 \eta^2/64$. }
This completes the proof of  Proposition~\ref{clm:coloring}. 
\end{proof}
\end{proof}

\section{Approximation algorithm  for the MAX-3-EQUAL problem}\label{sec:rounding}
In this section, we give a SDP based approximation algorithm for MAX-3-EQUAL  whose performance matches the hardness result from the last section. In particular, we prove the following theorem.
\begin{theorem}\label{thm:algo}
There is a polynomial time approximation algorithm for the MAX-3-EQUAL problem which achieves the following approximation ration : 
$$
\inf_{\delta \in (0,1)}  \frac{ 1 -   \frac{3\cos^{-1} (1-\delta)}{2\pi}}{1-\frac{3\delta}{4}} \approx 0.796.
$$
\end{theorem}
Thus, this theorem shows that  we have an approximation algorithm whose performance ratio matches the Unique Games hardness for this problem. Towards proving Theorem~\ref{thm:algo}, 
we  state a SDP relaxation for the MAX-3-EQUAL problem followed by a rounding procedure and then analyze the performance of this algorithm. The SDP formulation is essentially the generic SDP by Raghavendra~\cite{Rag:08} specialized to the MAX-3-EQUAL problem.  We assume that the variables are $x_1, \ldots, x_n \in \{-1,1\}$.  The constraint set {is} $E \subseteq [n]^3 \times \{-1,1\}^3$ such that for every $(i,j,k) \times (\eta_i, \eta_j, \eta_k) \in E$ we have a constraint that $\eta_i x_i = \eta_j x_j = \eta_k x_k$. {In other words, $\eta_i$ represents the polarity with which the variable $x_i$ appears in the constraint $E$ (likewise for $\eta_j$ and $\eta_k$). } The SDP relaxation is given in Figure~\ref{fig:algo}. 
\begin{figure}[tb]
\hrule
\vline
\begin{minipage}[t]{0.98\linewidth}
\vspace{10 pt}
\begin{center}
\begin{minipage}[h]{0.95\linewidth}
{\small
\underline{\textsf{SDP formulation}}

\begin{enumerate}
\item $\forall$ $i \in [n]$, $v_i \in \mathbb{R}^n$ and $\Vert v_i \Vert_2=1$.  
\item $\forall$  $i,j,k \in [n]^3$, $i<j<k$, $ $ $\alpha_{(i,j,k)}$, $\beta_{(i,j,k)}$, $\gamma_{(i,j,k)}$, $\delta_{(i,j,k)} \in \mathbb{R}^+ \cup \{0\}$ such that 
$$
\alpha_{(i,j,k)}+\beta_{(i,j,k)}+\gamma_{(i,j,k)}+\delta_{(i,j,k)} =1.
$$
\item $\forall$  $i,j,k \in [n]^3$, $i<j<k$, 
\begin{eqnarray*}
&& \alpha_{(i,j,k)}+\beta_{(i,j,k)}- \gamma_{(i,j,k)}-\delta_{(i,j,k)}  =\braket{v_i,v_j}. \\
&& \alpha_{(i,j,k)}-\beta_{(i,j,k)}+ \gamma_{(i,j,k)}-\delta_{(i,j,k)}  =\braket{v_j,v_k}. \\
&& \alpha_{(i,j,k)}-\beta_{(i,j,k)}- \gamma_{(i,j,k)}+\delta_{(i,j,k)}  =\braket{v_i,v_k} .\\
\end{eqnarray*}
\item For $e \in E$, where $e = (i,j,k) \times (\eta_i,\eta_j,\eta_k)$, define \begin{equation*}
\lambda(e)= 
\begin{cases} \alpha_{(i,j,k)} & \text{if $\eta_i = \eta_j = \eta_k$,}
\\
\beta_{(i,j,k)} &\text{if $\eta_i = \eta_j = -\eta_k$,} \\
\gamma_{(i,j,k)} &\text{if $-\eta_i = \eta_j = \eta_k$,}\\
\delta_{(i,j,k)} &\text{if $\eta_i = -\eta_j = \eta_k$,}\\
\end{cases}
\end{equation*}
\vspace{10pt}
\item Subject to the above, Maximize $\mathop{\mathbf{E}_{(i,j,k) \times (\alpha, \beta, \gamma) \in E}} [\lambda(e)] $.
\end{enumerate}

\vspace{5 pt}
}
\end{minipage}
\end{center}
\end{minipage}
\hfill \vline
\hrule
\caption{SDP relaxation for MAX-3-EQUAL problem}
\label{fig:algo}
\end{figure}
\begin{remark}
We note that Zwick~\cite{Zwick98} describes a SDP relaxation and a similar rounding procedure for the MAX-3-EQUAL problem. The paper also gives numerical evidence towards showing
that the performance ratio of their algorithm   is approximately $0.796$. However, the paper notes that they do not have an analytical proof of this and to the best of our knowledge, no analytical proof has appeared ever since.  We analyze a slightly different SDP and analytically show that the performance of it is indeed what we claim. 
There are a couple of differences between our SDP formulation and Zwick's SDP. The first one is that we use some additional real variables. However, this difference is purely cosmetic as the presence of those variables does not make our relaxation any tighter than Zwick's SDP. The second difference between our SDP and Zwick's SDP is that our SDP implies an additional set of constraints, namely,  for all $i,j, k \in [n]^3$, $1+\braket{v_i,v_j}+\braket{v_i,v_k}+\braket{v_j,v_k} \ge 0$. We should mention that this family of constraints appears in \cite{Zwick98} for SDP relaxations of some other CSPs but it is unclear from the paper if Zwick uses these constraints in the SDP relaxation for MAX-3-EQUAL as well.  Potentially, these additional constraints make our SDP tighter than that of Zwick though the reason we use these additional constraints is that our analysis becomes simpler. \newa{As remarked earlier, after the publication of the preprint,  Williamson~\cite{Will:13} pointed out to us that an analysis similar to ours had already appeared in \cite{GW95} in the context of MAX-DICUT. Using this analysis as a black-box, we can shorten the analysis of the rounding algorithm substantially. We however keep our original analysis here so that the paper is self contained.}
\end{remark}

To see why the SDP in Figure~\ref{fig:algo}  is a relaxation, consider a particular assignment to the variables $x_1, \ldots, x_n$. Let us define $\mathbf{v_0} \in \mathbb{R}^n$ as having $1$ in the first coordinate and $0$ everywhere else. If $x_i=1$, set $v_i = \mathbf{v_0}$. Else, if $x_i=-1$, set $v_i = -\mathbf{v_0}$. The rest of the variables are set as follows. For every triple $(i,j,k)$, $i<j<k$, 
\begin{itemize}
\item If $x_i = x_j = x_k$, then $\alpha_{(i,j,k)}=1$, $\beta_{(i,j,k)} = \gamma_{(i,j,k)} = \delta_{(i,j,k)}=0$. 
\item If $x_i = x_j = -x_k$, then $\beta_{(i,j,k)}=1$, $\alpha_{(i,j,k)} = \gamma_{(i,j,k)} = \delta_{(i,j,k)}=0$.
\item If $-x_i = x_j = x_k$, then $\gamma_{(i,j,k)}=1$, $\alpha_{(i,j,k)} = \beta_{(i,j,k)} = \delta_{(i,j,k)}=0$.
\item  If $x_i = -x_j = x_k$, then $\delta_{(i,j,k)}=1$, $\alpha_{(i,j,k)} = \beta_{(i,j,k)} = \gamma_{(i,j,k)}=0$.
\end{itemize}
It is easy to verify that with these assignments of $\alpha_{(i,j,k)}$, $\beta_{(i,j,k)}$, $\gamma_{(i,j,k)}$, $\delta_{(i,j,k)}$ and $v_i$, constraints $1$, $2$ and $3$ are indeed satisfied. Further, for this assignment, if a constraint $e \in E$ is satisfied, then it is easy to see that $\lambda(e)=1$. Also, if a constraint $e$ is not satisfied, then $\lambda(e)=0$. Thus, the objective value of the program for this assignment is exactly the fraction of constraints $e \in E$ which are satisfied and hence its a relaxation. 


\subsection{Rounding algorithm}\label{subsec:rounding}

Our rounding algorithm is as follows:
Let $\Sigma  \in \mathbb{R}^{n \times n}$ be the matrix such that $\Sigma_{i,j} = \braket{v_i,v_j}$. Note that $\Sigma$ is positive semidefinite.  
So, we let $\mathcal{X} \sim \mathcal{N}(\mathbf{0},\Sigma)$, i.e., $\mathcal{X}$ be a jointly normal distribution in $\mathbb{R}^n$ with mean at the origin and the covariance matrix $\Sigma$.  The rounding algorithm gets a sample $\mathcal{X}$ and assigns $x_i=1$ if $\mathcal{X}_i \ge 0$ and $-1$ otherwise. Here $\mathcal{X}_i$ denotes the $i^{th}$ coordinate of $\mathcal{X}$.  We will call this rounding as the ``random gaussian" rounding. We now prove Theorem~\ref{thm:algo} by analyzing the performance of this rounding algorithm. 

We would also like to remark that (perhaps not too surprisingly), if instead of the ``random gaussian" rounding, we would have used ``random hyperplane" rounding, the performance of the algorithm would have been the same and our analysis would have also gone through without any changes. 
\begin{proof}[Proof of Theorem~\ref{thm:algo}] 
We start by  considering a particular constraint $e \in E$.  {Without loss of generality, assume that $e= (i,j,k) \times (\eta_i, \eta_j,\eta_k)$, where $\eta_i = \eta_j = \eta_k=1$. We note that if the triple $(\eta_i,\eta_j, \eta_k)$ were to take some other value in $\{-1,1\}^3$, our analysis would remain unchanged.}

Now, for the particular edge $e$, its contribution to the SDP objective is $\lambda(e) = \alpha_{(i,j,k)}$. On the other hand, let  the expected contribution to the true objective from this edge be $\kappa(e)$. Note that
\begin{equation}\label{eq:perf}
\kappa(e) = \Pr[(\mathcal{X}_i, \mathcal{X}_j,\mathcal{X}_k \ge 0) \cup (\mathcal{X}_i, \mathcal{X}_j,\mathcal{X}_k < 0)].
\end{equation}
It is obvious that  the performance ratio of the algorithm is lower bounded by $\inf \kappa(e)/\lambda(e)$.  Hence, we will simply aim to prove a lower bound on $\inf \kappa(e)/\lambda(e)$. 
 Observe that for any $(i,j,k)$, $\alpha_{(i,j,k)}+ \beta_{(i,j,k)}+\gamma_{(i,j,k)}+\delta_{(i,j,k)}=1$. 
 Now, using this and plugging Fact~\ref{fac:ac}  into (\ref{eq:perf}), we get (below, we use $\alpha$ as a shorthand for $\alpha_{(i,j,k)}$ and likewise for $\beta$, $\gamma$ and $\delta$), \begin{eqnarray*}
\kappa(e) &=& 1 - \frac{ \cos^{-1}(\braket{v_i,v_j}) + \cos^{-1}(\braket{v_j,v_k})  +\cos^{-1}(\braket{v_i,v_k})}{2\pi}  \\
&=& 1 - \frac{ \cos^{-1}(2(\alpha+ \beta)-1) + \cos^{-1}(2(\alpha+ \gamma)-1) +\cos^{-1}(2(\alpha+ \delta)-1)}{2\pi}.
\end{eqnarray*}
Thus, for $a,b,c,d \in \mathbb{R}^+ \cup \{0\}$, if we define 
$$
g(a,b,c,d) {=}  \frac{1 - \frac{\cos^{-1} (2(a+b)-1) +\cos^{-1} (2(a+c)-1) +\cos^{-1} (2(a+d)-1)}{2\pi}}{a},
$$
\begin{eqnarray*}
\mathrm{then,  }\quad \frac{\kappa(e)}{\lambda(e)} \ge \inf_{a,b,c,d} g(a,b,c,d) \textrm{ subjected to} \ a+b+c+d=1 \textrm{ and } a,b,c,d \ge 0.
\end{eqnarray*}
For the purposes of the analysis, it is helpful to fix the value of $a$, and then find the optimum choice of $b$, $c$, $d$ for that value of $a$ to minimize $g(a,b,c,d)$. Subsequently, one optimizes over the choice of $a$. In other words, let us define $h_a(b,c,d)$ as 
$$
h_a(b,c,d) =\cos^{-1} (2(a+b)-1) +\cos^{-1} (2(a+c)-1) +\cos^{-1} (2(a+d)-1).
$$
\begin{eqnarray}\label{eq:boundary}
\Psi(a)=\sup_{b,c,d} h_a(b,c,d) \textrm{ subjected to} \ b+c+d=1-a  \textrm{ and } b,c,d \ge 0 \textrm{ where } a>0.
\end{eqnarray}
Hence, we now get that
\begin{equation}\label{eq:imp}
\frac{\kappa(e)}{\lambda(e)} \ge \inf_{0<a \le 1} \frac{1-\frac{\Psi(a)}{2\pi}}{a}.
\end{equation}
Thus, we now focus on finding $\Psi(a)$ for every $a \in (0,1]$. 
In order to find out $\Psi(a)$, we find out the local minima by evaluating the partial derivatives  of the function $h_a(b,c,d)$ and also investigate the value of $h_{a}(b,c,d)$ at the boundaries of the domain.  \newa{Williamson~\cite{Will:13} noted to us that  the expression  $ \inf_{0<a \le 1} \frac{1-(\Psi(a)/2\pi)}{a}$ had already been analyzed in  \cite{GW95} (see Lemma~7.3.2). However, we keep our original analysis here. }

\subsection{Supremum of $h_a(b,c,d)$ at the boundary of the domain: } The next claim gets the supremum of $h_a(b,c,d)$ when $b,c,d$ lie on the boundary of the domain defined 
in Equation~\ref{eq:boundary}. 
\begin{claim}\label{clm:boundary}
The supremum of  $h_a(b,c,d)$ when $b$, $c$ and $d$ lie on the boundary of the domain defined in (\ref{eq:boundary}) is $\cos^{-1}(2a-1) + 2 \cos^{-1}(a)$. 
\end{claim}
\begin{proof}
Note that because $b,c,d \ge 0$ and $b+c+d = 1-a$, the boundary of the domain is defined by at least one of $b$, $c$ and $d$ being $0$. Without loss of generality, we assume $b=0$. {Note that because $a$ is fixed, we are viewing the domain as a two dimensional object.}
In that case, $$h_a(0,c,d) =\cos^{-1}(2a-1) + \cos^{-1} (2(a+c)-1) +\cos^{-1} (2(a+d)-1),$$ with $c + d = 1-a$ and $c,d \ge 0$.  Performing the substitution $d=1-a-c$, we get 
\begin{equation}\label{hac}h_a(0,c,d)=  \cos^{-1}(2a-1) +\cos^{-1} (2(a+c)-1) +\cos^{-1} (1-2c),
\end{equation}  where $0 \le c \le 1-a$.   
Now, note that since $a$ is fixed,  $h_a(0,c,d)$ is solely a function of $c$. Hence, to find out the supremum of $h_a(0,c,d)$, we evaluate it at the end points of the domain, i.e., at $c=0$, $c=1-a$ and at its critical points. 

\begin{itemize}
\item If $c=0$, then $d=1-a$. Hence, at this point, $h_a(b,c,d) = h_a(0,0,1-a)=  \cos^{-1}(2a-1)+ \cos^{-1}(2a-1) + \cos^{-1}(1) = 2\cos^{-1}(2a-1)$. 
\item If $c=1-a$, then $d=0$. Hence, at this point, $h_a(b,c,d) = h_a(0,1-a,0)=  \cos^{-1}(2a-1)+ \cos^{-1}(1) + \cos^{-1}(2a-1) = 2\cos^{-1}(2a-1)$. 
\end{itemize}
  
  Having evaluated $h_a(0,c,d)$ at the boundary points, we now find out the critical points of this function.  Differentiating the expression in (\ref{hac}), we get 
  \begin{eqnarray*}\frac{\partial h_a(0,c,d)}{\partial c} = \frac{-2}{\sqrt{1-(2(a+c)-1)^2}} + \frac{2}{\sqrt{1-(1-2c)^2}} =0 .
\end{eqnarray*}
This implies that 
$$
1-(2(a+c)-1)^2 = 1-(1-2c)^2 
$$
$$
\Rightarrow (2(a+c) -1) = \pm (1-2c).
$$
This means that either $a=0$ or $a+2c=1$. Since $a>0$, we can neglect the first condition. Thus, the only condition we need to consider is $a+2c=1$. Because $a+c+d=1$, this means that 
$c = d =(1-a)/2$. Thus, $h_a(0,c,d) = \cos^{-1}(2a-1) + 2 \cos^{-1}(a)$. 
Thus, we get that 
\begin{equation}\label{eq:crit}
\sup_{c,d} h_a(0,c,d) = \sup \{ \cos^{-1}(2a-1) + 2 \cos^{-1}(a), 2 \cos^{-1}(2a-1) \} = \cos^{-1}(2a-1) + 2 \cos^{-1}(a).
\end{equation}
The last equality uses Fact~\ref{fac:trig1}.
\end{proof}
\subsection{Evaluation of $h_a(b,c,d)$ at the critical points:} The next claim evaluates the supremum of $h_a(b,c,d)$ at the critical points of the domain. 
\begin{claim}\label{clm:critical}
The supremum of $h_a(b,c,d)$ at the critical points inside the domain defined in (\ref{eq:boundary}) is given by 
$$
  \sup h_a(b,c,d) =\begin{cases}\pi + \cos^{-1}(4a-1) & \text{if $0 \le a \le 1/4$,}
\\
  3  \cos^{-1}((4a-1)/3) &\text{if $1/4 < a \le 1$.} \\ \end{cases}
  $$
\end{claim}
\begin{proof}
Note that $b+c+d=1-a$. Thus, we get 
$$
h_a(b,c,d) = \cos^{-1} (1-2c-2d) +\cos^{-1} (2(a+c)-1) +\cos^{-1} (2(a+d)-1).
$$
As $a$ is fixed, $h_a(b,c,d)$ is a function of $c$ and $d$ alone. At the critical point, 
$$
\frac{\partial h_a(b,c,d)}{\partial c} =  \frac{2}{\sqrt{1-(1-2c-2d)^2}} - \frac{2}{\sqrt{1-(1-2a-2c)^2}} =0,
$$
  $$
\frac{\partial h_a(b,c,d)}{\partial d} =  \frac{2}{\sqrt{1-(1-2c-2d)^2}} - \frac{2}{\sqrt{1-(1-2a-2d)^2}} =0.
$$
  Thus, at the critical point, 
  $$
  (1-2c-2d)^2 = (1-2a-2c)^2 = (1-2a-2d)^2
  $$
  $$
  \Rightarrow \pm  (1-2c-2d) = \pm (1-2a-2c)= \pm(1-2a-2d).
  $$
  We now solve for $c,d$ for the various possibilities listed above. 
  \begin{itemize}
  \item $1-2c-2d = 1-2a-2c = 1-2a-2d$. In this case, we get $a=c=d$ and hence $b=1-3a$. Since $b \ge 0$, this possibility occurs only when $0 \le a \le (1/3)$.  If this indeed holds,  
  $$h_a(b,c,d) = \cos^{-1}(1-4a) + \cos^{-1}(4a-1) + \cos^{-1}(4a-1) = \pi + \cos^{-1}(4a-1).$$ 
  \item $1-2c-2d = -(1-2a-2c) = 1-2a-2d$. In this case, we get $a=c=b$ and $d=1-3a$. Again as $d \ge 0$, this possibility occurs only when $0 \le a \le (1/3)$. As before, 
  $$
  h_a(b,c,d) = \cos^{-1}(1-4a) + \cos^{-1}(4a-1) + \cos^{-1}(4a-1) = \pi + \cos^{-1}(4a-1).
  $$
  \item $1-2c-2d = 1-2a-2c = -(1-2a-2d)$. This goes exactly the same way as in the previous case. Here again, we have
  $$
  h_a(b,c,d) = \cos^{-1}(1-4a) + \cos^{-1}(4a-1) + \cos^{-1}(4a-1) = \pi + \cos^{-1}(4a-1).
  $$
  \item $-(1-2c-2d) = 1-2a-2c = 1-2a-2d$. In this case, $b=c=d = (1-a)/3$. Now, we get 
  $$
  h_a(b,c,d)=\cos^{-1}((4a-1)/3) + \cos^{-1}((4a-1)/3) + \cos^{-1}((4a-1)/3) = 3  \cos^{-1}((4a-1)/3) .
  $$
  \end{itemize}
  Hence at the critical points, we have $$\sup h_a(b,c,d) =\begin{cases}  \sup \{\pi + \cos^{-1}(4a-1) ,  3  \cos^{-1}((4a-1)/3) \} & \text{ if $0 < a \le 1/3$,} \\
   3  \cos^{-1}((4a-1)/3) &\text{ if $a>1/3$.} \\\end{cases}$$
  However, using Fact~\ref{fac:trig2}, the above simplifies to saying that at the critical points, 
  $$
  \sup h_a(b,c,d) =\begin{cases}\pi + \cos^{-1}(4a-1) & \text{if $0 \le a \le 1/4$,}
\\
  3  \cos^{-1}((4a-1)/3) &\text{if $1/4 < a \le 1$.} \\
\end{cases}  $$
\end{proof}
Define $\zeta$ to be the smallest of the following three quantities:
$$
\left\{ \inf_{a \in (0,1/4]}  \frac{1 - \frac{\pi + \cos^{-1}(4a-1)}{2\pi}}{a} , \inf_{a \in (1/4,1]}  \frac{1 - \frac{3 \cos^{-1}((4a-1)/3)}{2\pi}}{a},\inf_{a \in (0,1]}  \frac{1 - \frac{2 \cos^{-1}(a) + \cos^{-1}(2a-1)}{2\pi}}{a}\right\}.
$$
Combining Claims~\ref{clm:boundary} and \ref{clm:critical} along with (\ref{eq:imp}), we get that 
$
\frac{\kappa(e)}{\lambda(e)} \ge \zeta$. 
 By Fact~\ref{fac:trig3}, the first quantity inside the definition of $\zeta$ simplifies to $1$ as follows: 
 \begin{equation}\label{eq:bound}
 \inf_{a \in (0,1/4]}  \frac{1 - \frac{\pi + \cos^{-1}(4a-1)}{2\pi}}{a} = \frac{1 - \frac{\pi + \cos^{-1}(4 \cdot (1/4)-1)}{2\pi}}{(1/4)} =1.
 \end{equation}
 At this point, we are left with the task of finding the following quantities: 
 $$
 \inf_{a \in (0,1]}  \frac{1 - \frac{2 \cos^{-1}(a) + \cos^{-1}(2a-1)}{2\pi}}{a} \quad \inf_{a \in (1/4,1]}  \frac{1 - \frac{3 \cos^{-1}((4a-1)/3)}{2\pi}}{a}.
 $$
 Thus, we are now left with the task of finding the infimum of two single-variable functions and then taking the minima of these two quantities. We do this computation by evaluating these two functions at sufficiently many points and then taking the infimum of these. For a mathematical justification, see Appendix~\ref{app:precision}. Doing the numerical computation, we get, 
  \begin{equation}\label{eq:eval1}
 \inf_{a \in (0,1]}  \frac{1 - \frac{2 \cos^{-1}(a) + \cos^{-1}(2a-1)}{2\pi}}{a}  =   [0.803125,0.803325].
 \end{equation}
 \begin{equation}\label{eq:eval2}
 \inf_{a \in (1/4,1]}  \frac{1 - \frac{3 \cos^{-1}((4a-1)/3)}{2\pi}}{a} = [0.795970,0.796170].
 \end{equation}
 Further, the value of $a$ achieving the infimum in (\ref{eq:eval2}) is $a=0.700296 \pm 0.000001$. Hence, we have that 
 $$
 \frac{\kappa(e)}{\lambda(e)} \ge  \inf_{a \in (1/4,1]}  \frac{1 - \frac{3 \cos^{-1}((4a-1)/3)}{2\pi}}{a} =  \inf_{a \in (0,1]}  \frac{1 - \frac{3 \cos^{-1}((4a-1)/3)}{2\pi}}{a}.
 $$
 The second equality (i.e., making the domain $(0,1]$ instead of $(1/4,1]$) follows because
  Fact~\ref{fac:trig2} and (\ref{eq:bound}) can be combined  as: 
 $$
\forall 0< a \le 1/4 \quad \frac{1 - \frac{3 \cos^{-1}((4a-1)/3)}{2\pi}}{a} \ge \frac{1 - \frac{\pi + \cos^{-1}(4a-1)}{2\pi}}{a} \ge 1.
 $$
 Put $\delta = 4(1-a)/3$. Then, we get that  
 $$
  \inf_{a \in (0,1]}  \frac{1 - \frac{3 \cos^{-1}((4a-1)/3)}{2\pi}}{a} = \inf_{0 \le \delta < 4/3} \frac{1 - \frac{3 \cos^{-1}(1-\delta)}{2\pi}}{1 -\frac{3\delta}{4}} =\inf_{0 \le \delta \le 1} \frac{1 - \frac{3 \cos^{-1}(1-\delta)}{2\pi}}{1 -\frac{3\delta}{4}}  .$$
 Here the last equality is true because we have earlier observed that the infimum of the expression in (\ref{eq:eval2})  is obtained when $a \approx 0.700$. This means the corresponding value of $\delta \approx  0.400 <1$.   Thus, making the domain of $\delta$ to be $(0,1]$ instead of $(0,4/3]$ does not affect the value of the infimum. This also conclude the proof of the theorem. 
 
 \end{proof}
 
 \section{Difficulty in getting optimal results for MAX-k-EQUAL}\label{sec:difficult}
 
 Given our results on MAX-3-EQUAL, a very obvious question is whether or not our results can be extended to MAX-k-EQUAL for $k>3$. More concretely, since it is known that assuming the UGC, Raghavendra's SDP achieves the optimal approximation ratio for every CSP, it is natural to ask if the ``random gaussian" rounding algorithm described in subsection~\ref{subsec:rounding} also achieves this ratio. We now explain the difficulty in proving such a  result in a nutshell.

 Consider the case of MAX-k-EQUAL. Let  $(g_1, \ldots, g_k)$ be jointly normally distributed random variables such that each $g_i \sim \mathcal{N}^n(0,1)$ (the value of $n$ is immaterial as long as $n \ge k$). Assume that for all $1 \leq a \leq n$, the covariance matrix of $g_1(a),\ldots ,g_k(a)$ is given by $\rho \in \mathbb{R}^{k \times k}$ (and $\rho$ is independent of $a$). Here $g_i(a)$ represents the $i^{th}$ coordinate of $g_i$. 
 
 Let us define a family of distributions $\mathcal{D}(\rho)$ {over $\{-1,1\}^k$} in the following way:  $\mathcal{A} \in \mathcal{D}(\rho)$ if and only if  \begin{itemize} \item $\forall$ $1 \le i \le k$, $\mathcal{A}(i)$ is a uniformly random bit. \item $\forall$ $1\le i <j \le k$, $\mathbf{E} [\mathcal{A}(i) \cdot \mathcal{A}(j)] =\rho_{ij}$. \end{itemize} We next define the following two quantities:
 \begin{eqnarray*}
&& h_s(\rho) = \Pr_{ g \in \mathcal{N}^n(0,1)} [\forall i \in [k] \ \ g \cdot g_i \ge 0] +\Pr_{ g \in \mathcal{N}^n(0,1)} [\forall i \in [k] \ \ g \cdot g_i < 0]. \\
 &&
 h_c(\rho) = \max_{\mathcal{A} \in \mathcal{D}(\rho)} [\mathcal{A}(1) = \ldots = \mathcal{A}(k)] .
 \end{eqnarray*}
It is easy to show that the approximation ratio achieved by the ``random gaussian" rounding algorithm on Raghavendra's SDP is (lower)-bounded by $\inf_{\rho} h_s(\rho)/h_c(\rho)$. 

If we want to show that the ``random gaussian" rounding algorithm on Raghavendra's SDP indeed achieves the optimal approximation ratio (assuming the UGC), then the task essentially boils down to constructing a dictatorship test for MAX-k-EQUAL whose ratio of soundness to completeness is $\inf_{\rho} h_s(\rho)/h_c(\rho)$. To do this, let us assume that $\arg \inf_{\rho} h_s(\rho)/h_c(\rho) = \rho'$. 

First of all, we  construct a dictatorship test for MAX-k-EQUAL whose completeness is $h_c(\rho')$. 
To do this, let us assume that the distribution (in $\mathcal{D}(\rho')$) which achieves the maximum in the definition of $h_c(\rho')$ is $\mathcal{A}$. 
The dictatorship test is as follows: Given a function $f : \{-1,1\}^n \rightarrow \{0,1\}$,  we sample $(X_1, \ldots, X_k) \in \mathcal{A}^n$ and accept if and only if $f(X_1) = \ldots =f(X_k)$. 

  It is easy to see that if $f$ is a dictator, then the probability that $f(X_1) = \ldots =f(X_k)$ is exactly $h_c(\rho')$. Thus, the completeness of the dictatorship test is exactly $h_c(\rho')$. 
  The hard part is to bound the soundness of the dictatorship test. In other words, assuming that $f$ is a balanced function where every coordinate has a low-influence, we need to bound the probability that 
$f(X_1) = \ldots = f(X_k)$. An application of the invariance principle~\cite{Mossel:10}  says that it suffices to bound the following quantity: Let $f' : \mathbb{R}^n \rightarrow \{0,1\}$ be a function on the gaussian space such that $\mathbf{E} [f'(x)] =1/2$. Let $g_1, \ldots, g_k \sim \mathcal{N}^n (0,1)$ be jointly normally distributed random variables where for all $1\le a \le n$, the covariance matrix of $g_1(a), \ldots, g_k(a)$ is given by  $\rho'$. We need to upper bound the probability that $f'(g_1) = \ldots = f'(g_k)$. The result in \cite{IM12} says that as long as all the off-diagonal entries of $\rho'^{-1}$ are non-positive, the probability is maximum when $f'$ is a halfspace. However, if $f'$ is indeed a halfspace, then $\Pr[f'(g_1) = \ldots = f'(g_k)] = h_s(\rho')$. Thus, if all the off-diagonal entries of $\rho'^{-1}$ are non-positive, then the soundness of the dictatorship test is $h_s(\rho')$.

 For the case of $k=3$, by a direct analysis of the rounding algorithm, we showed that $\rho'$ is a matrix who diagonal entries are all $1$ and all the off-diagonal entries are the same positive quantity. From this, it is easy to check that all the off diagonal entries of $\rho'^{-1}$ are non-positive and hence the results of \cite{IM12} are applicable here. On the other hand, for $k>3$, it seems difficult to compute $\rho'$ exactly or even prove that all the off diagonal entries of $\rho'^{-1}$ are non-positive. This makes it impossible to apply the results of \cite{IM12} here. One might consider the possibility of doing computer simulations to find $\inf_{\rho} h_s(\rho)/h_c(\rho)$ for $k>3$ (or to make a reasonable conjecture about this quantity). However, note that for $k>3$,  $h_c(\rho)$ is not even completely determined by $\rho$. As a result, even doing computer simulations for $k>3$ is rather complicated. This summarizes the difficulty in extending our results to MAX-k-EQUAL for $k>3$. 


 \section{Conclusion}
Our results illustrate the importance of Gaussian partition results in establishing exact optimal UGC hardness and rounding schemes. Not only did we show that a new Gaussian partition result allows to obtain exact UGC hardness of MAX-3-EQUAL, we also showed  how the trivial Gaussian partition gives near optimal hardness for MAX-k-CSPs. 

There are many interesting open problems that emerge from our work and previous work. 
Perhaps the most natural open problem is regarding the hardness of MAX-k-EQUAL. 
In particular, is it true that the generic SDP from \cite{Rag:08} followed by the random gaussian / hyperplane rounding is optimal for MAX-k-EQUAL (assuming the Unique Games Conjecture)? 

A more general challenge it to obtain further optimal Gaussian partition results. In particular we recall the Standard Simplex Conjecture from \cite{IM12} which says that if $(X,Y)$ are jointly normal random variables in $\mathbb{R}^n$ such that $X,Y \sim \mathcal{N}^n(0,1)$ and $\Cov(X,Y) = \rho I_n$ where $\rho>0$, then a partitioning of the gaussian space into $k$ parts of equal measure such that $(X,Y)$ fall in the same partition is maximized when the partition corresponds to a $k$-simplex centered at the origin. Proving this, will have consequences for hardness of MAX-k-CUT.

\appendix

\section*{Acknowledgements} 
We are grateful to Ori Gurel-Gurevich and Ron Peled for answering several questions related to \cite{BGP12}. 
We are grateful to Per Austrin and Jelani Nelson for helpful comments on an earlier draft.  
We also thank Anand Bhaskar and Piyush Srivastava for help with using Mathematica. \newa{We are grateful to David Williamson for letting us know that our analysis of the MAX-3-EQUAL SDP is essentially identical to the analysis of the MAX-DICUT SDP from \cite{GW95}. }

AD is grateful to Luca Trevisan and Madhur Tulsiani for numerous discussions and Satish Rao for financially supporting him during the period when this work was done.  
\bibliography{allrefs}

  \section*{APPENDIX}
  \setcounter{section}{1}
  \subsection{Useful Trigonometric facts}~\label{sec:missing-proofs}
\begin{fact}\label{fac:trig1}
For every $0 \le a\le 1$, $2 \cos^{-1}(a) - \cos^{-1}(2a-1) \ge 0$. 
\end{fact}
\begin{proof}
Note that \begin{eqnarray*}
\cos (2 \cos^{-1}(a)) = 2a^2 - 1 \le 2a -1 = \cos ( \cos^{-1}(2a-1)).
\end{eqnarray*}
Now recall that if $0 \le \theta, \phi \le \pi$, then $\cos \theta \le \cos \phi$ if and only if $\theta \ge \phi$. Clearly, as $a \ge 0$, $0 \le 2 \cos^{-1}(a)  \le \pi$. Also, $0 \le \cos^{-1}(2a-1) \le \pi$. This concludes the proof. 
\end{proof}
\begin{fact}\label{fac:trig2}
Let $-1 \le x \le 1$.  Then, if $x \ge 0$, then $\pi + \cos^{-1} (x) \le 3 \cos^{-1}(x/3)$.   Else if, $x \le 0$, then $\pi + \cos^{-1} (x) \ge 3 \cos^{-1}(x/3)$. 
\end{fact}
\begin{proof}
Consider $f(x) =  3 \cos^{-1}(x/3) - \pi - \cos^{-1} (x)$.  Then, note that within the domain $(-1,1)$, the function is differentiable and hence
$$
\frac{d f(x)}{dx} = \frac{-1}{\sqrt{1- x^2/9}} + \frac{1}{\sqrt{1-x^2}}.
$$
{It is easy to see that for all $x \in (-1,1)$, $df(x)/dx \ge 0$. As $f(0)=0$, we can conclude that for all $x \in [-1,0]$, $f(x) \le 0$ and for all $x \in [0,1]$, $f(x) \ge 0$.} This concludes the proof. 
\end{proof}
\begin{fact}\label{fac:trig3}
Let $f:(0,1/4] \rightarrow \mathbb{R}$ be defined as
$$
f(x) = \frac{1 - \frac{\pi + \cos^{-1} (4x-1)}{2\pi}}{x}.
$$
Then, $f(x)$ is decreasing in the interval $(0,1/4]$. 
\end{fact}
\begin{proof}
We do a change of variables. Put $\cos \theta = 4x-1$. Thus proving the claim is equivalent to showing that  for $\pi/2 \le \theta \le  \pi$,  $g(\theta)$ (defined below) is an increasing function in the said interval. 
$$
g(\theta) = 4 \cdot \frac{\frac12 - \frac{\theta}{2 \pi}}{1 + \cos \theta}.
$$
Next, we evaluate $g'(\theta)$. 
$$
g'(\theta) = 4 \cdot \frac{(1+ \cos \theta) \cdot \frac{-1}{2\pi} + \sin \theta \cdot \left( \frac12 - \frac{\theta}{2 \pi}\right) }{(1 + \cos \theta)^2}.
$$
Note that if we show $g'(\theta) \ge 0$ in the interval $\theta \in [\pi/2, \pi]$, then it implies that $g(\theta)$ is an increasing function in the same interval. Thus, we need to show that for  $\theta \in [\pi/2, \pi]$
$$
(1+ \cos \theta) \cdot \frac{-1}{2\pi} + \sin \theta \cdot \left( \frac12 - \frac{\theta}{2 \pi}\right)  \ge 0. 
$$
{Using the identities $1+\cos \theta = 2\cos^2 (\theta/2) $ and $\sin \theta = 2 \cos (\theta/2) \cdot \sin (\theta/2)$, we get }
$$
(\pi -\theta) \sin (\theta/2) \ge \cos (\theta/2) \quad \Longleftrightarrow \quad \pi -\theta - \cot (\theta/2)  \ge 0.
$$
So, we finally need to show that $h(\theta) =  \pi -\theta - \cot (\theta/2) $ is non-negative in the interval $\theta \in [\pi/2,\pi)$.  But $h'( \theta) = - \cot^2 (\theta) <0$. This means that 
$h(\theta) \ge h(\pi) =0$ proving our claim. 
\end{proof}
\begin{fact}\label{fac:ineq2}
For $0 \le x \le 1$, $\cos^{-1} (x) \le \pi/2 - x$.
\end{fact}
\begin{proof}
\begin{eqnarray*}
\sin x \le x \quad \Rightarrow \quad \cos(\pi/2- x) \le x \quad \Rightarrow \quad \pi/2- x \ge \cos^{-1}(x).
\end{eqnarray*}
\end{proof}
\begin{fact}\label{fac:ineq1}
For $0 \le x \le 1$, $\cos^{-1} (x-1) \le \pi -\sqrt{x}$.
\end{fact}
\begin{proof}
Let $g(x) = \cos (\sqrt{x}) -1 +x$. Observe that $g(0) =0$. Also, 
$$
g'(x) = -\frac{\sin \sqrt{x}}{2 \sqrt{x}} +1 >0.
$$
This implies that $g(x) \ge 0$ for all $0 \le x \le 1$. This implies
\begin{eqnarray*}
\cos (\sqrt{x}) -1 +x \ge 0 \quad &\Rightarrow& \quad   x-1 \ge  - \cos (\sqrt{x}) = \cos (\pi -\sqrt{x}) \\
 \quad &\Rightarrow& \quad \cos^{-1}(x-1) \le \pi -\sqrt{x}.
\end{eqnarray*}
\end{proof}
\begin{fact}\label{fac:ineq3}
For $0 \le x \le 1$, $\cos^{-1}(x) \le 3 \sqrt{1-x}$. 
\end{fact}
\begin{proof}
Put $x = 1- \epsilon$. Then, the claim is equivalent to proving that for $0 \le \epsilon \le 1$, $\cos^{-1}(1-\epsilon) \le 3\sqrt{\epsilon}$. 
Towards this, define $g(\epsilon) = 3 \sqrt{\epsilon} - \cos^{-1}(1-\epsilon) $. Clearly, $g(0) =0$. Next, we note that 
$$
g'(\epsilon) = \frac{3}{2\sqrt{\epsilon}} -\frac{1}{\sqrt{1-(1-\epsilon)^2}}=\frac{3}{2\sqrt{\epsilon}} -\frac{1}{\sqrt{2\epsilon - \epsilon^2}} =\frac{3}{\sqrt{4\epsilon}} -\frac{1}{\sqrt{2\epsilon - \epsilon^2}} .
$$
It is easy to see that for  $\epsilon\in [0,1]$, $g'(\epsilon) \ge 0$. Hence, for $\epsilon\in [0,1]$, $g(\epsilon) \ge 0$ finishing the proof.\end{proof}
\begin{fact}\label{fac:ineq4}
For $0.9 \le x \le 1$, $\cos^{-1}(2x-1) \le 5 \sqrt{1-x}$. 
\end{fact}
\begin{proof}
Note that putting $x=1-\epsilon$, this is equivalent to proving that for $0 \le \epsilon \le 0.1$, $\cos^{-1}(1-2\epsilon) \le 5 \sqrt{\epsilon}$. To prove this, consider the function 
$g(\epsilon) = 5 \sqrt{\epsilon}-\cos^{-1}(1-2\epsilon)$. Clearly, $g(0)=0$. Also, 
$$
g'(\epsilon) = \frac{5}{2\sqrt{\epsilon}} -\frac{2}{\sqrt{1-(1-2\epsilon)^2}} =\frac{5}{2\sqrt{\epsilon}} - \frac{1}{\sqrt{\epsilon - \epsilon^2}}.
$$
Now, note that for $\epsilon\in [0,0.1]$, $g'(\epsilon) \ge 0$. Hence, for $\epsilon\in [0,0.1]$, $g(\epsilon) \ge 0$ finishing the proof.
\end{proof}

\subsection{Justification for numerically finding the minima}\label{app:precision}
In Section~\ref{sec:rounding}, we numerically evaluate the minimum of two single variable  functions using the software ``Mathematica". We now give a detailed explanation of how we find the minima of these functions to the desired error and the mathematical soundness of this computer-assisted procedure. 

\subsubsection{Infimum of $h_1(a)$}
Given the function $h_1: (0,1] \rightarrow \mathbb{R}$ from Section~\ref{sec:rounding} (which is defined as)
$$
h_1(a) =\frac{1 - \frac{2 \cos^{-1}(a) + \cos^{-1}(2a-1)}{2\pi}}{a}.
$$
To find $\inf_{a \in (0,1]} h_1(a)$, we do the following: 
\begin{itemize}
\item Show that for the interval  $A_1 = (0,x_s]$ and $A_2= [x_t,1]$ (where $x_s= 0.179$ and $x_t =0.99$), $\inf_{x \in A_1} h_1(x) \ge 0.85$ and $\inf_{x \in A_2} h_1(x) \ge 0.83$.
\item Show that for the interval $A_3= (x_s,x_t)$, and $x \in A_3$, $|h_1'(x)| \le \Delta$ where $\Delta=500$.
\item Divide the interval $A_3$ into $\Delta/\eta$ (with $\eta=10^{-4}$) intervals of equal length and evaluate $h_1$ at each of these points where $h_1(a)$ is evaluated at each point with an error of $\epsilon = 10^{-6}$. Subsequently, take the minimum of all these numbers. 
\end{itemize}
It is clear that the above procedure returns the infimum of $h_1$ in the interval $(0,1]$ to within error $\epsilon + \eta/2 \le 10^{-4}$.  Following this procedure, $\inf_{a \in (0,1]} h_1(a)$ was obtained to be $0.803225$.  Since, we note that the error can be at most $10^{-4}$, hence  $\inf_{a \in (0,1]} h_1(a) \in [0.803125,0.803325]$. 

We now give proofs for the first and the second item in the above procedure. 
\begin{proposition}\label{prop:11}
Let  $h_1 : [0,1] \rightarrow \mathbb{R}$ be defined as $$h_1(a) =\frac{1 - \frac{2 \cos^{-1}(a) + \cos^{-1}(2a-1)}{2\pi}}{a}.$$
Then, for $0 \le a \le 0.179$, $h(a) \ge 0.85$. 
\end{proposition}
\begin{proof}
Using Fact~\ref{fac:ineq2} and Fact~\ref{fac:ineq1}, we have
$$
h_1(a) =\frac{1 - \frac{2 \cos^{-1}(a) + \cos^{-1}(2a-1)}{2\pi}}{a} \ge \frac{2a + \sqrt{2a}}{2 \pi a} = \frac{1}{\pi} + \frac{1}{\pi \sqrt{2a}}.
$$
Plugging in the values, this implies that as long as $a \le 0.179$, $h_1(a) \ge 0.85$.
\end{proof}
\begin{proposition}\label{prop:12}
Let  $h_1 : [0,1] \rightarrow \mathbb{R}$ be defined as $$h_1(a) =\frac{1 - \frac{2 \cos^{-1}(a) + \cos^{-1}(2a-1)}{2\pi}}{a}.$$
Then, for $0.99 \le a \le 1$, $h(a) \ge 0.83$. 
\end{proposition}
\begin{proof}
Using Fact~\ref{fac:ineq3} and Fact~\ref{fac:ineq4}, we have
$$
h_1(a) =\frac{1 - \frac{2 \cos^{-1}(a) + \cos^{-1}(2a-1)}{2\pi}}{a} \ge \frac{1 - \frac{6 \sqrt{1-a} + 5 \sqrt{1-a}}{2\pi}}{a}.
$$
Plugging in the values, this implies that as long as $0.99 \le a \le 1$, $h_1(a) \ge 0.83$.
\end{proof}
Proposition~\ref{prop:11} and Proposition~\ref{prop:12} imply the proof of the first item. The next proposition implies the correctness of the third item. 
\begin{proposition}
For every $a \in [0.179,0.99]$, $|h_1'(a)| \le 500$. 
\end{proposition}
\begin{proof}
 $$
 h_1'(a) = \frac{\frac{a}{\pi \sqrt{1-a^2}} + \frac{a}{\pi \sqrt{1-(1-2a)^2}}-1 + \frac{\cos^{-1}(a)}{\pi} + \frac{\cos^{-1}(2a-1)}{2\pi}}{a^2}.
 $$
This implies that
$$
|h_1'(a)| \le  \frac{\frac{a}{\pi \sqrt{1-a^2}}+ \frac{a}{2\pi \sqrt{a-a^2}}+3}{a^2} \le \frac{3}{a^2} +  \frac{\frac{1}{\pi \sqrt{1-a^2}}+ \frac{1}{2\pi \sqrt{a-a^2}}}{a}.
$$
To bound the value of $|h_1'(a)|$, we consider the two cases: when $0.179 \le a \le 0.5$ and when $0.99 \ge a>0.5$. Splitting into these two cases, it is easy to show
$$
|h_1'(a)| \le 500.
$$
 \end{proof}
 \subsubsection{Infimum of $h_2(a)$}
 Recall that we need to find the following quantity: 
 $$
\inf_{a \in (1/4,1]} h_2(a) \quad \textrm{where} \quad h_2(a) =  \frac{1 - \frac{3 \cos^{-1}((4a-1)/3)}{2\pi}}{a}. $$
We do the following change of variables: We put $(4a-1)/3 = \cos x$. Then, the problem becomes finding the quantity 
$$
\inf_{x \in [0,\pi/2)} g(x) \quad \textrm{ where } \quad g(x) =4 \cdot \frac{1 - \frac{3x}{2\pi}}{1+3 \cos x} .
$$
To find $\inf_{x \in [0,\pi/2)} g(x)$, we do the following: 
\begin{itemize}
\item Show that for $x \in [0,\pi/2)$, $|g'(x)| \le \Delta$ where $\Delta=50$.
\item Divide the interval $[0,\pi/2)$ into $\Delta/\eta$ (with $\eta=10^{-4}$) intervals of equal length and evaluate $g(x)$ at each of these points where $g(x)$ is evaluated at each point with an error of $\epsilon = 10^{-6}$. Subsequently, take the minimum of all these numbers. 
\end{itemize}
It is clear that the above procedure returns the infimum of $h_2$ in the interval $(0,1]$ to within error $\epsilon + \eta \cdot  (\pi/4) \le 10^{-4}$.  Following this procedure, $\inf_{a \in (0,1]} h_2(a)$ was obtained to be $0.796070$.  Since the error is bounded by $10^{-3}$, we know $\inf_{a \in (0,1]} h_2(a) \in [0.795970,0.796170]$.
We now give proof for the first item in the above procedure. 
\begin{proposition}\label{prop:13}
Let  $g : [0,\pi/2) \rightarrow \mathbb{R}$ be defined as above. 
Then, for $x \in [0,\pi/2)$, $|g'(x)| \le 50. $
\end{proposition}
\begin{proof}
$$
g'(x) =12 \cdot \frac{\sin x - \frac{3 x \sin x}{2\pi} - \frac{1}{2\pi}  -\frac{3 \cos x}{2\pi}}{(1+3 \cos x)^2}.$$
It is now trivial to see that the absolute value of $g'(x)$ is bounded by $50$ at all points in $[0,\pi/2)$.
\end{proof}

\end{document}